\documentclass[twocolumn,10pt]{IEEEtran}
\usepackage{bm,cite,algorithm,float,amsmath,amssymb}

\usepackage{graphicx,epstopdf}
\usepackage{pifont}

\usepackage[justification=centering]{caption}

\usepackage{bbm}

\makeatletter

\newcommand{\Rmnum}[1]{\expandafter\@slowromancap\romannumeral #1@}
\makeatother
\usepackage{mathrsfs}

\usepackage{amsfonts}

\usepackage{textcomp}
\usepackage{color}

\usepackage[noend]{algpseudocode}

\usepackage{graphicx}
\usepackage{subfigure}

\usepackage{balance}
\usepackage[utf8x]{inputenc}
\IEEEoverridecommandlockouts

\IEEEoverridecommandlockouts
\usepackage{mathtools}
\usepackage{amsthm}

\newtheorem{theorem}{Theorem}

\newtheorem{corollary}{Corollary}

\usepackage{dsfont}

\newtheorem{remark}{Remark}

\newtheorem{lemma}{Lemma}

\usepackage{epsfig}

\begin{document}
	
	\title{{Analysis of Random Access in NB-IoT Networks with Three Coverage Enhancement Groups: A Stochastic Geometry Approach}}
	
	\author{
		\IEEEauthorblockN
		{ Yan Liu, ~\IEEEmembership{Student Member,~IEEE,}
			Yansha Deng,
			~\IEEEmembership{Member,~IEEE,}
			Nan Jiang,
			~\IEEEmembership{Student Member,~IEEE,}\\
			Maged Elkashlan,~\IEEEmembership{Senior Member,~IEEE,}
			and Arumugam Nallanathan ~\IEEEmembership{Fellow,~IEEE}
		}\\
	
	\vspace{-0.2cm}
		
		\thanks{Y. Liu, N. Jiang, M. Elkashlan, and A. Nallanathan are with School of Electronic Engineering and Computer Science, Queen Mary University of London, London, UK
			(e-mail:\{yan.liu, nan.jiang, maged.elkashlan, a.nallanathan\}@qmul.ac.uk). }
		\thanks{Y. Deng is with Department of Engineering, King's College London, London, UK (Corresponding author: Yansha Deng (e-mail:yansha.deng@kcl.ac.uk).)}

		\thanks{{Part of this work was presented in IEEE Global Communications Conference, Dec. USA 2019 \cite{9013330}.
		}}
			
	}

	\maketitle

	\begin{abstract}
		NarrowBand-Internet of Things (NB-IoT) is a new 3GPP radio access technology designed to provide better coverage for Low Power Wide Area (LPWA) networks. 
		To provide reliable connections with extended coverage, a repetition transmission scheme and up to three Coverage Enhancement (CE) groups are introduced into NB-IoT during both Random Access CHannel (RACH) procedure and data transmission procedure, where each CE group is configured with different repetition values and transmission resources.
		To characterize the RACH performance of the NB-IoT network with three CE groups, this paper develops a novel traffic-aware spatio-temporal model to analyze the RACH success probability, where both the preamble transmission outage and the collision events of each CE group jointly determine the traffic evolution and the RACH success probability.
		Based on this analytical model, we derive the analytical expression for the RACH  success probability of a randomly chosen IoT device in each CE group over multiple time slots with different RACH schemes, including baseline, back-off (BO), access class barring (ACB), and hybrid ACB and BO schemes (ACB\&BO).
		{Our results have shown that  the RACH success probabilities of the devices in three CE groups  outperform that of a single CE group network but  not for all the groups,   which is affected by the choice of the categorizing parameters.}This mathematical model and analytical framework can be applied to evaluate the performance of multiple group users of other networks with spatial separations.
		
		


	\end{abstract}
	
	\vspace*{+0.3cm}
	

	\section{Introduction}
	The Internet of Things (IoT) offers a wide spectrum of opportunities for innovative applications designed to improve our life quality.
	The plethora of opportunities offered by IoT services include health-care, automation, metering, tracking, monitoring, and etc\cite{7123563}\cite{8058399}, in which ubiquitous connectivity and coverage among massive number of IoT devices are required for successful operation of these IoT services.
	Cellular-based network is deemed as one solution to provide connectivity for massive number of IoT devices, due to its advantages in high scalability, diversity, and security, as well as low cost without additional infrastructure deployments \cite{name2015}\cite{6375894}.
	
	There exist several challenges in cellular-based IoT networks, including low device cost (below 5 USDs), limited uplink latency (below 10s), massive number of devices (up to 40 per household), long battery life (10 years), and enhanced coverage (20dB better than GPRS)
	\cite{shariatmadari2015machine}\cite{landstrom2016nb}.
	To cope with these challenges, the Third Generation Partnership Project (3GPP) has standardized the NB-IoT in Release 13, 
	which defines narrow transmission bandwidth, repetition transmission, single-tone transmission, enhanced discontinuous reception, power spectral density (PSD) boosting , and other network architectural updates \cite{landstrom2016nb}\cite{schlienz2016narrowband}. 

	Coverage Enhancement (CE) is one of the features proposed for NB-IoT networks, which can be achieved with the help of the narrower carrier bandwidth and the repetition transmission \cite{el2017m2m}.
	On one hand, NB-IoT can provide a higher PSD with respect to Long-Term Evolution (LTE)\cite{huawei2018}, as LTE operates in physical resource block (PRB) units of 180 kHz, but the NB-IoT can operate with 15 kHz and 3.75 kHz\cite{7794567}.
	On the other hand, RACH repetition and data repetition are enforced in both uplink and downlink for coverage enhancement.
	More importantly, according to the 3GPP standard \cite{name2015}, to support various traffic with
	different coverage conditions, each Base Station (BS)  categorizes its IoT devices into up to three CE groups, which provides efficient management of a massive number of IoT devices depending on their received signal quality.
	The RACH repetition value is determined by the BS based on the CE group of the IoT device through the RACH procedure \cite{schlienz2016narrowband}\cite{8385556}.

	In NB-IoT, the main purpose of RACH procedure is to achieve uplink synchronization  and  obtain  the  grant for  initial  access to the network\cite{dahlman20134g}, in which the first step is to transmit a RACH preamble.
	Notably, massive connections in NB-IoT may bring simultaneous RACH requests under limited number of available preambles. 
	Thus, it is of great importance to model and analyze the RACH performance of NB-IoT networks, which can be useful for system design and optimization. 
	{
		In  \cite{luo1999stability,duan2013dynamic,7486114,8385148}, mathematical models of contention-based RACH focusing on the Signal-to-Interference-plus-Noise Ratio (SINR) outage or collision problem have been studied.
		The authors in \cite{7486114} combined queueing theory and stochastic geometry to analyze the stability region in a discrete-time slotted RACH network. In \cite{8385148}, the authors designed a 
		RACH protocol for the standalone Long-Term Evolution (LTE) system in
		an unlicensed spectrum (SA LTE-U), where the UEs are divided into several
		groups, and at any time only one group is activated and allowed for its UEs to send RA attempt, which avoids the inter group UEs’ collision. 
		Importantly, previous results in LTE systems cannot be directly applied to NB-IoT due to its unique characteristics, including transmission repetition, three CE groups configuration, frequency hopping,  and etc.
		The authors in \cite{8605340} investigated a tradeoff between repetition of preambles in NB-IoT and their retransmission for the RACH procedure in an
		NB-IoT system with single CE group. The capacity limits of RACH
		for LTE-based IoT and NB-IoT services were studied in \cite{6678832} and \cite{8239592}, respectively.
		Although \cite{8239592}  studied the random access channel in NB-IoT networks with three CE groups, it did not consider the repetition schemes in NB-IoT, the packets evolution, the time correlation interference and etc.
	}
	

	Our previous work \cite{8258982} has provided a general analytical framework to characterize the RACH success probability in NB-IoT networks with preamble repetition scheme based on the preamble transmission model in  \cite{nan2018random} and collision model in \cite{nan2018collision}.
	Note that \cite{8258982} only considered NB-IoT networks with a single CE group in a single time slot with the transmit power of the IoT device determined by the path-loss inversion power control due to the analytical simplicity, which does not align with the practical NB-IoT networks with multiple CE groups setting. 
	According to the 3GPP standard \cite{Tel2016Physical}, for the IoT device with the repetition value larger than two, its transmit power should be set as the cell specific maximum transmit power. 

	Different from \cite{8258982}, we model and analyze the RACH success probability taking into account the three geographically separated CE groups in each cell with their repetition values in NB-IoT networks  in multiple time slots.
	{
		We also evaluate the efficiency of several RACH schemes based on the presented analytical model, including baseline, back-off (BO), access class barring (ACB), and hybrid ACB and BO schemes (ACB\&BO), in the NB-IoT network to alleviate uplink congestion by reducing the high interference and high collision probability when massive IoT devices contend for the uplink channel resource at the same time\cite{3gpp2011study}\cite{6525600}. }
	{
		In this paper, we address the following fundamental questions:
		1) how to model the analyze the RACH success probabilities in the NB-IoT networks with three CE groups;
		2) to what extent the repetition transmission scheme improves the RACH success probabilities in different groups; 
		3) to what extent the RACH success probabilities of three CE groups outperform those of a single CE group;
		4) to what extent the ACB, BO, and hybrid  ACB$\&$BO  schemes improve the RACH success probabilities in different groups.}
	{
		To solve these problems, we develop a novel spatio-temporal mathematical framework to analyze and evaluate the RACH success probability for NB-IoT networks
		with three CE groups using stochastic geometry and probability theory, taking into account the SINR outage events as well as the collision events at the BS. }
	
	Generally speaking, in the NB-IoT network with three CE groups, the physical layer parameters and network topology can strongly affect the RACH performance of each CE group, due to that the received SINR distribution at the BS depends upon the joint distribution of the received powers from the serving IoT device and the interfering IoT devices in each CE group, which ultimately depends on the network topology.
	In this scenario, the random positions and the numbers of IoT devices in three CE groups make accurate modeling and analysis of the interference in each CE group even more complicated.
	
	Even though stochastic geometry has been regarded as a powerful tool to model and analyze mutual interference between transceivers in the wireless networks with its tractability and realism in modeling irregular node locations \cite{andrews2011tractable,deng2016physical,deng2016artificial,elsawy2013stochastic}, there are three aspects that limit the application of conventional stochastic geometry analysis to the RACH performance analysis in NB-IoT networks with three CE groups over multiple time slots: 1) conventional stochastic geometry works focused on analyzing normal uplink and downlink data transmission channel, where the intra-cell interference is not considered, due to the ideal assumption that each orthogonal sub-channel is not reused in a cell, which is not the case when massive IoT devices in each CE group of a cell may randomly choose and transmit the same preamble using the same sub-channel, bringing the intra-cell interference; 
	2) the interference field in conventional stochastic geometry works is mostly modeled by a homogeneous PPP to maintain tractability, which is not the case for the interference field in each CE group of each cell with spatial separation into three coverage areas among three CE groups;
	3) most existing stochastic geometry works always consider inversion power control for analytical simplicity, as the radius term is missing from the desired received power term.
	
	According to the 3GPP standard, the consideration  of each CE group is different and we need to model and analyse each CE group separately and differently. 
	The new challenges of this work are listed as: 
	{
		1) both the intra- and inter- group interference for the same group is considered, due to that the IoT devices in the same group in a cell may randomly choose and transmit the same preamble using the same sub-channel;
		2) the interference field of each CE group needs to be modeled separately based on their different received power region; 
		3) the transmit powers of CE group 1 and 2 are generally a fixed power, and thus the interference from interfering IoT devices depends on the different and random transmission distances in each CE group; 
		4) the configured parameters of three CE groups are different and related, which determines the definition equation of RACH success probability; 
		5) our analysis considering multiple time slots need capture the traffic change over time due to new arrival packets, and previous unsuccessful packets.}


	
	The contributions of this paper can be summarized as follows: 
	
	{
		1) We present a {novel spatio-temporal mathematical framework } for analyzing RACH access in the NB-IoT network with three CE groups using stochastic geometry and probability theory. 
		In the spatial domain, stochastic geometry is applied to model and analyze the mutual interference for each CE group. In the time domain, probability theory is applied to model the correlation of the buffer state and the transmission state over different time slots
	}

	{2) Based on the framework, we propose a tractable approach to analyze contention-based RACH success probability of IoT devices in each CE group for different RACH schemes, including baseline, BO, ACB and hybrid ACB\&BO schemes. We first derive the exact expression for the RACH success probability of a randomly chosen IoT device in each CE group in a single time slot and then extend the analysis to multiple time slots for different RACH schemes by considering preamble transmission policy and queue evolution.
	}
	

	3) We develop a realistic simulation framework to capture the randomness locations, preamble transmission as well as the real packets arrival, accumulation, and departure of each IoT device in each time slot and verify our derived RACH success probability of the IoT device in each CE group.
	
	4) Our numerical results  presented in this paper can be applied the performance evaluation of multiple group users of other networks with spatial separations.


	The rest of the paper is organized as follows. 
	Section II  presents a  system
	model. 
	Section III derives the RACH success probability of a randomly chosen IoT device in each CE group in a single time slot. Section IV derives the RACH success probability of a randomly chosen IoT device in each CE group over multiple time slots with different RA schemes.
	Our results and simulations are described in Section V. Finally, Section VI  has drawn the conclusion.

	\section{System Model}
	We consider a traffic-aware uplink spatio-temporal model for NB-IoT networks with configuring three repetition parameters for three  CE  groups  in a cell where multiple IoT devices simultaneously start their RACH procedure after receiving a group paging message. 
	In the spatial domain, BSs and IoT devices are spatially distributed following two independent Poisson Point Processes\footnote{{Our work assumes that BSs are distributed following PPP like most of the stochastic geometry works to present a general and tractable framework for RACH analysis in the NB-IoT networks that focus on the massive connectivity. This is different from the work \cite{7218400} considering that the BSs are deployed according to cell planning in the finite networks with finite nodes.}} (PPPs) $\Phi_{\rm{B}}$ and $\Phi_{\rm{D}}$ with intensities $\lambda_{\rm{B}}$ and $\lambda_{\rm{D}}$, respectively.
	In the temporal domain, the packets arrival at each IoT device in each time slot is modeled as independent Poisson arrival process ${\Lambda _{New}}$ with intensities ${\varepsilon _{New}}$\cite{7886285}\cite{6477828}.
	Following \cite{7486114}\cite{nan2018random}\cite{7886285}, the time is slotted into discrete time slots, and the IoT devices and the BSs  remain spatially static once they are deployed.
	{Following \cite{8258982}\cite{chiu2013stochastic}, we assume each IoT device associates to its geographically nearest BS, where a Voronoi tessellation is formed. 
		Moreover, we consider additive noise with average power $\sigma ^2$ and a Rayleigh fading with the channel power gain $h$ assumed to be exponentially distributed  with unit mean, i.e., $h\sim$ Exp(1). All channel gains are assumed to be independent and identically distributed (i.i.d.) in space and time.
	}

	\begin{figure*}[htbp!]
	\centering
	\includegraphics[width=6.5in,height=2in]{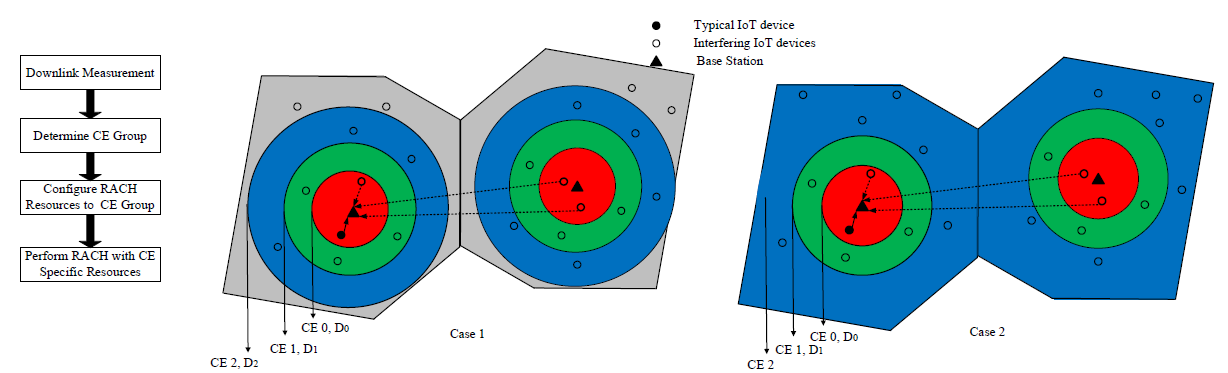}
	\caption{NB-IoT CE Groups 
	}
	\label{fig:1}
\end{figure*}	
	\subsection{Problem Statement}
	As shown in Fig. 1, the IoT devices are divided into three CE groups (i.e., CE group $i$, $i$ = 0, 1 and 2) according to their downlink RSRP measurement as further discussed in Section II.D.	
	A packet can only be transmitted via the NarrowBand Physical Uplink Shared CHannel (NPUSCH), which can be scheduled by the associated BS after the active IoT device executing a RACH to request uplink channel resources with the BS as shown in Fig. 2.
	\begin{figure}[htbp!]
		\centering
		\includegraphics[width=3.6in,height=3.6in]{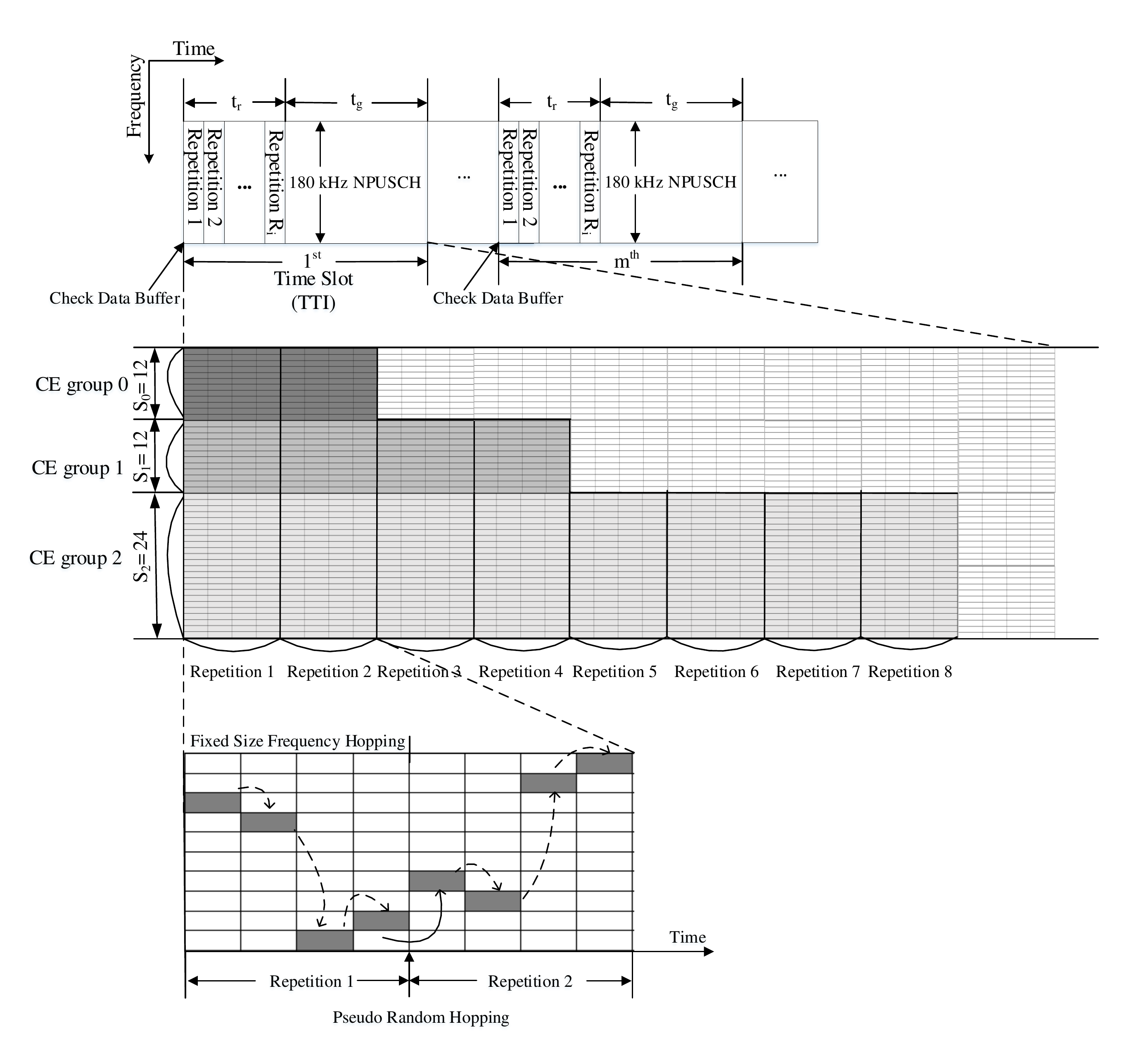}
		\caption{Structure of RACH Procedure 
		}
		\label{fig:2}
	\end{figure}
	As only active IoT devices execute the contention-based RACH procedure to establish a connection with the network, we need to derive the active probability of the IoT device at the beginning of each Transmission Time Interval (TTI).
	Here, the active IoT device represents that an IoT device is with non-empty buffers and without access restriction, which will be detailed in Section III.
	Thus, we need to derive the non-empty probability $\mathcal{A}_i^m$  and the non-restrict probability $\mathcal{R}_i^m$  of the IoT device in the $m$th TTI for CE group $i$.
	As only IoT devices that has performed successful RACH transmit packets, we need to derive the RACH success probability ${\cal P}_i^m$ of the IoT device in the $m$th TTI for CE group $i$.
	In order to {analyze} the time-slotted contention-based RACH in the NB-IoT network with three CE groups, we assume that the actual intended packet transmission is always successful (i.e., the data transmission success probability is one) if the corresponding RACH succeeds. 
	Note that the data transmission after a successful RACH can be extended following the analysis of RACH success probability. 
	Here, we limit ourselves to focus on the impact of repetition scheme and CE groups to RACH procedure.

	\subsection{Random Access Procedure}
	
	The contention-based RACH procedure consists of four steps, where a randomly selected preamble is transmitted to the associated BS on NB-IoT Physical Random Access CHannel (NPRACH), for a given number
	of times (i.e., the dedicated repetition value) in step 1, and control information with the BS is exchanged in step 2,3,4\cite{dahlman20134g}\cite{nan2018collision}.
	The RACH requests from massive connections in NB-IoT 
	simultaneously under limited number of available preambles is one of the main challenges, thus we focus on the contention of preamble in step 1 of contention-based RACH with the assumption that  {the steps} 2,3,4 of RACH are always successful whenever the step 1 is successful following\cite{nan2018collision}.
	That is to say a RACH procedure is always successful if the IoT device successfully transmits the preamble to its associated BS. 
	In this case, the RACH success is determined by two reasons: 1) the preamble being successfully transmitted to the associated BS (i.e., received SINR is greater than the SINR threshold $\gamma_{th}$); and 2) no collision occurs (i.e., no other IoT devices successfully transmit the same preamble to the typical BS simultaneously).
	It is known that collision in step 1 of RACH can be detected by the BS, when the collided IoT devices are separable in terms of the power delay profile \cite{dahlman20134g}. 
	Our model follows the assumption of collision handling in \cite{nan2018collision}\cite{3gpp2011study}, where collision events are detected by the BS after it decodes the preambles in step 1 of RACH; hence, the BS will not send the RAR and the IoT device can not proceed to the next step of RACH procedure and need to restart the RACH procedure in the next available RACH opportunity\cite{7510814}.
	
	\subsection{Physical Random Access CHannel}
	As shown in Fig. 2, in the NPRACH, a preamble consists of four symbol groups transmitted without gaps on a single subcarrier and can be repeated several times using the same transmit power.
	The subcarrier spacing of NPRACH is 3.75 kHz and up to 48 subcarriers can be allocated to NPRACH. 
	These sub-carriers are exclusively shared by three CE groups with a basic sub-carrier allocation unit of 12 sub-carriers \cite{Tel2016}. 
	Current 3GPP standard mandates the number of subcarriers in each CE group to be configured as a multiple of 12, with maximal value of 48 \cite{8385556}\cite{7785170}.
	According to whitepaper\cite{schlienz2016narrowband},
	frequency hopping is applied on symbol group granularity, i.e. each symbol group is transmitted on a different subcarrier,
	where the first preamble symbol group is transmitted via a subcarrier selected via the pseudo-random hopping (i.e., the hopping depends on the current repetition time and the Narrowband physical Cell ID, a.k.a NCellID\cite{schlienz2016narrowband}), and the following three preamble symbol groups are transmitted via subcarriers determined by the fixed size frequency hopping \cite{Tel2016} (i.e. each symbol group is transmitted on a different subcarrier) as shown in Fig. 2. 
	That is to say, if two or more IoT devices choose the same first subcarrier in a single RACH opportunity, the following subcarriers (i.e., in the same RACH opportunity) would be same, due to that these two hopping algorithms lead to one-to-one correspondences between the first subcarrier and the following subcarriers (i.e., these IoT devices either collide on the full set or not collide at all in a single RACH opportunity). 

	\subsection{CE Group Determination}
	As shown in Fig. 1, the IoT device determines its CE group by measuring the downlink RSRP.
	In this paper, we use the Signal-to-Noise Ratio (SNR) thresholds following\cite{7569029}.
	In the following subsections, we describe and formulate the coverage area of each CE group,  the preamble set as well as repetition value, the density, and the uplink transmit power of IoT devices in each CE group.

	\subsubsection{Coverage Area of Each CE Group}
	According to \cite{name2015}, 
	the BS uses the constant power $P_{\rm DL}$ to broadcast the Downlink Control Information (DCI) signal to all the IoT devices in its own cell.
	Based on the received SNR of DCI signal measured at each IoT device and the SNR thresholds \{$\delta_{1}$, $\delta_{2}$\}, each IoT device independently determine its associated CE group following the rule below:
	\begin{align}\label{AREA}
	\hspace{-0.5cm}\begin{cases}
	\displaystyle\frac{{{P_{\rm DL}}{x^{ - \alpha }}}}{\omega } \ge {\displaystyle\delta _{1}}, &\mbox{device belongs to CE group 0},\\
	{\displaystyle\delta _{2}} \le \displaystyle\frac{{{P_{\rm DL}}{x^{ - \alpha }}}}{\omega } < {\delta _{1}},&\mbox{device belongs to CE group 1},\\
	\displaystyle\frac{{{P_{\rm DL}}{x^{ - \alpha }}}}{\omega } < {\delta _{2}},&\mbox{device belongs to CE group 2,}
	\end{cases}
	\end{align}
	where $\omega$ is the noise power at the IoT device and $x$ is the IoT device’s distance from the BS. 
	The devices with the lowest received powers (less than $\delta_{2}$) belong to the group 2 and the BS need to allocate higher repetition value to this group; the ones with the highest received powers (more than $\delta_{1}$) belong to the group 0 and the BS can allocate lower repetition value to this group to allow fairness performance among three CE groups. 
	
	It is worth noticing that $\delta_{1}$ and $\delta_{2}$ depend on the particular modulation and coding scheme (MCS) used by the BS to broadcast $i$th CE group DCI, as well as the expected QoS level and propagation environment. Thus, the maximum distance $D_i$ between the BS and an IoT device belonging to the $i$th CE group can be derived from (\ref{AREA}) as
	\begin{align}\label{T_i_2}
	\begin{cases}  
	{D_0} = {\left( {\displaystyle\frac{{{\delta _{1}}\omega }}{P_{DL}}} \right)^{ - (1/\alpha )}},\\
	{D_1} = {\left( {\displaystyle\frac{{{\delta _{2}}\omega }}{P_{DL}}} \right)^{ - (1/\alpha )}}.
	\end{cases}
	\end{align}
	
	Specifically, we define Coverage Area (CA${_i}$) as the area in which the CE group $i$ IoT devices located in.
	As shown in Fig. 1, CA${_0}$ is represented by a circle centered at the BS with radius $D_0$; CA${_1}$ is represented by a annulus centered at the BS with internal radius $D_0$
	and external radius $D_1$; CA${_2}$ is represented by a annulus centered at the BS with internal radius $D_1$ and external radius $D_2$, where $D_0$, $D_1$ are given in \eqref{T_i_2}
	and $D_2$ is given following \cite{6516885} as
	\begin{align}\label{D2}
	D_2= 1/\sqrt {\pi {\lambda _B}}.
	\end{align}
	In consequence, an IoT device belongs to CE group $i$ if it is located in the coverage area CA${_i}$.

	\subsubsection{{Preamble Set and Repetition Value Configured for IoT devices in Each CE Group}}
	To serve IoT devices in three CE groups, the NB-IoT network can configure three NPRACH resource configurations for each CE group  in a cell separately.
	The BS will notify the NPRACH configuration to the IoT device in the system information by broadcasting, which include the preamble set and preamble repetition value required for the estimated CE group as well as preamble transmit power.
	According to the 3GPP standard\cite{name2015}\cite{schlienz2016narrowband}, 
	we set $S_i$ as the number of orthogonal subcarriers (preambles) reserved by the BS for CE group $i$ ( $S_0+S_1+S_2≤48$) with  configuration sets $\{S_0, S_1, S_2\}  \in \{ \{12,12,24\}, \{12,24,12\}, \{24,12,12\} \}$. Thus, each preamble in CE group $i$ has an equal probability
	$(1/{S_i})$ to be chosen.
	{
		IoT devices in CE group $i$ transmit the chosen preamble from set $S_i$ using the same transmit power for ${K_i}$  times, where
		the repetition value specified for each configuration can be chosen from the sets $K_0 \in \{1, 2\}$ and $K_1,K_2 \in \{4, 8, 16, 32, 64, 128\}$.}
	The preamble repetition value of the higher CE group is usually larger than that of lower CE group, i.e., $R_0<R_1<R_2$.
	
	\subsubsection{Uplink Transmit power of IoT devices in Each CE Group}
	
	Based on the 3GPP standard\cite{Tel2016}\cite{name2015}, in the uplink, the transmit power depends on a set of cell specific parameters and UE measured parameters.
	Specifically, the transmit power in CE group 0 is determined by the path-loss inversion power control, where each IoT device compensates for its own path-loss to keep the average received signal power equal to the same threshold $\rho$. 
	A standard power-law path-loss model is considered in CE group 0, where the path-loss attenuation is defined as $x^{-\alpha}$, with the propagation distance $x$ and the path-loss exponent ${\alpha}$.
	The transmit power in CE group 1 and 2 is generally a fixed power $P$ (the cell specific fixed transmit power on slot).
	Therefore, the transmit power of an IoT device in the CE group $i$ can be expressed as 
	\begin{align}\label{P_i}
	\begin{cases}  
	P_{0,j}=\rho (r_{0,j})^{\alpha}, & i\mbox{ = 0},\\
	P_{i,j}=P, & i\mbox{ = 1, 2},
	\end{cases} 
	\end{align}
	where $r_{0,j}$ is the distance from the $j$th IoT device in CE group 0 to the typical BS.

	\subsubsection{Density of IoT Devices in Each CE Group}
	{Note  that  IoT  devices  in  the  same CE group  may  choose the same preamble from the same preamble set $S_i$ during step 1 in RACH procedure, and only the IoT devices choosing the same preamble will generate interference\footnote{{As shown in Fig. 1, we consider intra-group interference, i.e., the interference from the IoT devices choosing the same preamble in the same group associated with the same BS. We also consider the inter-group interference, i.e., the interference from the IoT devices in other cells choosing the same preamble, due to that the IoT devices in different cells share the preamble sequence pool among BSs. In this work, each cell configures the same subcarrier set to each group. For example, group 0 in cell one and group 0 in cell two are configured with the same orthogonal subcarrier set $S_0$. That is to say, the configuration of each cell is the same as each other.}}.
	It is necessary to derive the density of  IoT  devices   choosing  the  same  preamble in  each CE  group.}
	Note that the spatial correlations among the interfering IoT devices on the aggregate interference are ignored \cite{7917340}. In fact, the exact locations and the mutual spatial correlations of the interfering IoT devices are of less significance to the SINR distribution at the BS. Instead, the density (number) of the IoT devices along with their relative locations with respect to the BS are the main contributions that affect the SINR.
	{We approximate the interfering devices of each CE group by the PPP ${\Phi _i}$ with the  density $\lambda_i$ in the following \textbf{Lemma 1} \cite{7557010}}.
	\begin{lemma}
		{(Approximation) We approximate the interfering devices of each CE group by the PPP ${\Phi _i}$ with the  density $\lambda_i$ given as }
		{
			\begin{align}\label{T_i}
			\begin{cases}  
			{\lambda _0} \approx  g_0{\lambda _D}=(1 - \exp ( - {\lambda _B}\pi D_0^2)){\lambda _D},\\
			{\lambda _1} \approx  g_1{\lambda _D}=(\exp ( - {\lambda _B}\pi D_0^2) - \exp ( - {\lambda _B}\pi D_1^2)){\lambda _D},\\
			{\lambda _2} \approx  g_2{\lambda _D}=(\exp ( - {\lambda _B}\pi D_1^2) - \exp ( - {\lambda _B}\pi D_2^2)){\lambda _D},  \\ \hspace{2.2in} (for\ Case\  \textit{1} ), \\ 
			{\lambda _2} \approx  g_2{\lambda _D}=(\exp ( - {\lambda _B}\pi D_1^2) ){\lambda _D},  \hspace{0.15in} (for\ Case\  \mbox{2}), 
			\end{cases} 
			\end{align}}
		where $g_0$, $g_1$, $g_2$ are thinning probabilities, $D_0$, $D_1$ are given in  \eqref{T_i_2} and $D_2$ is given in \eqref{D2}.
		{As the Voronoi cells do not have a constant radius, we consider two cases to analyze the CE group 2 respectively: Case 1, set  the external radius to  CE group 2 as $D_2$; Case 2, set the external radius of  CE group 2 equals to the Voronoi cell radius as shown in Fig. 1.}
	\end{lemma}
	\begin{proof}
		See Appendix A.
	\end{proof}


	The density of IoT devices in CE group $i$ choosing the same preamble can be expressed according to the thinning process as\cite{kingman1993poisson}
	\begin{align}\label{density}
	{\lambda _{i}^a} = {\lambda _i}/{S_i}.
	\end{align}
	
	\subsection{Traffic Model}
	We consider a time-slotted NB-IoT network, where the channel resource assignment of NPRACHs only occurs at the beginning of a TTI as shown in Fig. 3.
	According to the 3GPP standard \cite{schlienz2016narrowband}, the NPRACH happens at the beginning of a time slot within a small interval duration ${t _r}$, and the rest of a time slot is a gap interval duration ${t _g}$ for data transmission. 
	Without loss of generality, we assume that each IoT device is equipped with an infinite size buffer to store data packets received from higher layers. 
	We model the new arrived packets in the $m$th time slot $N_{New}^m$ at each IoT device as independent Poisson arrival process ${\Lambda _{New}}$ with intensities ${\varepsilon _{New}}$ as \cite{6477828}\cite{6376046}. 
	Therefore, the number of new packets $N_{New}^m$ in the $m$th time slot is described by the Possion distribution with $N_{New}^m \sim Pois(\mu _{New}^m)$, where $\mu _{New}^m = ({t _r} + {t _g})\varepsilon _{New}^m$.
	Packets are transmitted according to a First Come First Serve
	(FCFS) rule \cite{gow2006mobile} and a packet is dropped from the IoT device buffer once the RACH succeeds.
	Otherwise, the packet is kept in the buffer in the first place of the queue, and the IoT device will try to request channel resource for the packet in the next available RACH.
	Therefore, the number of accumulated packets in the $m$th time slot $N_{Cum}^m$ is evolved following transmission condition over time, which has been detailed and analyzed in our previous work \cite{nan2018random}.
	At the beginning of the NPRACH in each time slot, each IoT device needs to check its buffer status to determine whether itself requires to attempt RACH. 
	In detail, the buffer status is determined by the new arrived packets and the accumulated packets that unsuccessfully departs before the last time slot. 
	
	
	\begin{table}[htbp!]
		\centering
		\caption{Notation Table}
		{\renewcommand{\arraystretch}{1.2}
			\begin{tabular}{|p{1cm}|p{7cm}|}
				\hline
				
				$\lambda_B$ & The intensity of BSs \\ \hline 
				$\lambda_D$ & The intensity of IoT devices  \\ \hline
				${\varepsilon _{New}}$ & The intensity of new arrival packets \\ \hline
				$h$ & The Rayleigh fading channel power gain  \\ 
				\hline
				${r}$ & The distance between an IoT device and its associated BS \\ \hline 
				$\alpha$ & The path-loss exponent  \\ 
				\hline
				${K_i}$ & The RACH repetition value of CE group $i$ \\ \hline  
				$P_{DL}$ & The downlink transmit power  \\ 
				\hline
				$\delta_{1}$, $\delta_{2}$  & The Target SNRs \\ \hline 
				$\omega$ & The noise  power in the downlink\\ 
				\hline
				$CA_i$ &  The area of the CE group $i$ \\ \hline $D_i$ &  The radius of the $CA_i$  \\ 
				\hline
				$\lambda_i$ &  The intensity of IoT devices in CE group $i$  \\ \hline
				$N_i$ &  The number of {intra-group} interfering IoT devices in CE group $i$  \\ 
				\hline
				$S_i$ &  The number of available preambles in CE group $i$ \\ \hline 
				$\rho$ & The full path-loss power control threshold \\ 
				\hline				
				$P_{i,j}$ & The uplink transmit power of the device in CE group $i$  \\ \hline
				$\lambda _{i}^a$ & The average intensity of IoT devices using the same preamble in CE group $i$  \\ 
				\hline
				$\mu _{New}^t $ & The intensity of new arrival packets in the $t$th time slot \\ \hline
				$ N _{New}^t$ & The number of new arrived packets in the $t$th time slot\\
				\hline
				${\tau _r}$ & The PRACH duration \\ \hline
				${\tau _g}$ & The gap interval duration between two RACHs\\
				\hline	
				$Q_{ACB}$ & The ACB factor with the ACB scheme \\ \hline 
				$T_{BO}$ & The BO factor with the BO scheme  \\
				\hline	
				$\gamma _{th}$ & The SINR threshold \\ \hline
				$c$ & $c = 3.575$ is a constant \\
				\hline 
				$\mathcal{A}_i^t$ & The non-empty probability of each IoT device in the $t$th time slot for CE group $i$ \\ \hline 
				$\mathcal{R}_i^t$ & The non-restrict probability of each IoT device in the $t$th time slot for CE group $i$\\
				\hline
				$\mu_{Cum}^t$ & The intensity of accumulated packets in the $t$th time slot \\ \hline
				$ N_{New}^t$ & The number of accumulated packets in the $t$th time slot\\
				\hline
				$I_i$ & The aggregate interference for CE group $i$ \\ \hline 
				$\sigma^2$ & The noise power in the uplink \\
				\hline

			\end{tabular}
		}
		\label{table_accord}
	\end{table}
	\subsection{Transmission Schemes}

	In the NB-IoT network, a large number of IoT devices try to access the network simultaneously, which leads to a low RACH success probability and high network congestion due to mass concurrent data and signaling transmission\cite{3gpp2011study}.
	This may cause unexpected delays, packet loss, traffic overload, waste of radio resources, extra energy consumption, and even service interruption.
	In this case, efficient RACH transmission mechanisms are required for congestion reduction. 
	In this paper, we focus on evaluating and comparing the RACH performance of NB-IoT network via four RACH schemes:
	\subsubsection{Baseline scheme} According to \cite{7078932}, each IoT device attempt RACH immediately when there exits packets in the buffer. This is the simplest scheme without any control of traffic.
	\subsubsection{Access Class Barring (ACB) scheme} According to 3GPP standard \cite{3gpp2011study},  the ACB scheme has been standardized to prevent IoT devices from overloading RACH. 
	In ACB mechanism, initially a BS broadcasts an access barring factor ${\rm Q}_{ACB}$, which is specified by the BS according to the network condition\cite{dahlman20134g}\cite{3gpp2011study}. 
	When an IoT device initiates RACH, the device draws a random number $q \in [0,1]$, and compares this with ${\rm Q}_{ACB}$. 
	If $q < {\rm Q}_{ACB}$, the device is allowed to perform RACH procedure. 
	ACB scheme is a basic congestion control method that reduces RACH attempts from the side of IoT devices based on the ACB factor.
	\subsubsection{Back-off (BO) Scheme} BO scheme is introduced in 3GPP standard  \cite{3gpp2011study} to delay RACH attempts of IoT devices. 
	According to \cite{6525600}, each IoT device transmits packets the same as baseline scheme when there is no failure in the last time slot. 
	However, if a RACH fails in the $m$th time slot, the IoT device will perform the next RACH trial in the ($m+T_{BO}+1$)th time slot after a backoff period $T_{BO}$ time slots, where $T_{BO}$ is specified by the Backoff Indicator (BI).
	
	\subsubsection{ACB\&BO Scheme}
	The ACB\&BO scheme combined the ACB and BO schemes together. The BS first broadcasts the ACB factor $Q_{ACB}$, and then each active IoT
	device attempts a RACH with probability $Q_{ACB}$,
	i.e., each IoT device defers its RACH and waits for $T_{BO}$ time slots with probability ($1-Q_{ACB}$) if a RACH fails.
	
	The main notations of the proposed protocol are summarized in TABLE I.


	\section{General Single Time Slot Model}
	This section presents a general single time slot analytical model to characterize the  RACH  success probability of a randomly chosen IoT device in each CE group with different RACH schemes. 
	We formulate the RACH  success probability taking into account both the preamble outage and the collision. The RACH success probability ${\cal P}_i^1$ is defined as
	{
		\begin{align}\label{RACH}
		&{{\cal P}_i^1}={{\mathbb{E}}_N}{\Big[ { {{{\mathbb P}_{S,i,0}}[{K_i}]}{\prod\limits_{m = 1}^{n_i} {\Big(1 - {{\mathbb P}_{S,i,m}}[{K_i}]\Big)}\Big|N_i=n_i }} \Big]}\nonumber\\
		&=\sum\limits_{n_i = 1}^\infty  {\Big\{ {\underbrace {{\mathbb P}[N_i = n_i]}_{\rm I}\underbrace {{{{{\mathbb P}_{S,i,0}}[K_i}}]}_{{\rm I}{\rm I}}\underbrace {\prod\limits_{m =1}^{n_i} {\Big(1 - {{{\mathbb P}_{S,i,m}}[{K_i}}]\Big)}\Big|N_i=n_i }_{{\rm I}{\rm I}{\rm I}}} \Big\}}.
		\end{align}}
	Part \Rmnum{1} is the probability that the number of {intra-group} interfering IoT devices in CE group $i$ for a typical BS is equal to $n_i$, part \Rmnum{2} is the preamble transmission success probability of the typical IoT device in CE group $i$, and part \Rmnum{3} is the preamble transmission failure probability that the preambles transmitting from other $n_i$ {intra-group} interfering IoT devices in CE group $i$ are not successfully received by the BS,  i.e., the non-collision probability of the typical IoT device conditioning on $n_i$.

	{The  randomly  chosen IoT device transmits a preamble successfully if any repetition successes, and in a single repetition, a preamble is successfully received at the associated eNB if its all four received SINRs are above the SINR threshold $\gamma_{th}$. }
	Thus, the preamble transmission success probability of a randomly chosen IoT device in CE group $i$  under ${K_i}$ repetitions conditioning on $n_i$ number of {intra-group} interfering IoT devices is derived as
	\begin{align}\label{preamble1}
	{{{{\mathbb P}_{S,i,0}}[K_i}}] = 1 - \prod\limits_{{{k_i}} = 1}^{{K_i}} {\Big( {1 - {{\mathbb P}_{i,0}}[{\theta _{{k_i}}}| {r_{i,0}} ]} \Big)}, 
	\end{align}
	where $N_i=n_i$ is the number of {intra-group} interfering IoT devices in CE group $i$ (i.e., using the same preamble as the typical IoT device simultaneously in CE group $i$ in the same cell), {$r_{i,0}$ is the distance from the typical IoT device in CE group $i$ to its associated  BS}, and
	\begin{align}\label{FOUR SINR}
	{\theta _{{k_i}}} = \Big\{&{\rm{SINR}}_{{{k_i}}}^1 \ge {\gamma _{th}}, {\rm{SINR}}_{{{k_i}}}^2 \ge {\gamma _{th}},\\ \nonumber &{\rm{SINR}}_{{{k_i}}}^3 \ge {\gamma _{th}},{\rm{SINR}}_{{{k_i}}}^4\ge {\gamma _{th}} \Big\}.
	\end{align}
	In \eqref{FOUR SINR}, $\gamma _{th}$ is the SINR threshold, SINR$_{{{k_i}}}^1$, SINR$_{{{k_i}}}^2$, SINR$_{{{k_i}}}^3$, and SINR$_{{{k_i}}}^4$ are the received SINRs of the four continuous symbol groups in the ${k_i}$th repetition.
	
	Based on the Binomial theorem,  \eqref{preamble1} can be rewritten as \begin{align}\label{SI}
	{{{\mathbb{P}}_{S,i,0}}[{{K_i}}}]&= \sum\limits_{{{{k_i}}} = 1}^{{{{K_i}}}} {{( - 1)}^{{{{k_i}}} + 1}} {\Big( \begin{array}{l}
		{{{K_i}}}\\{{{k_i}}}\end{array}  \Big)}
	{{\mathbb{P}}_{i,0}}[{\theta _1},{\theta _2}, \cdots ,{\theta _{{{{k_i}}}}}| {r_{i,0}} ] ,
	\end{align}
	where $\Big( \begin{array}{l}
	{{K_i}}\\
	{{k_i}}
	\end{array} \Big) =\displaystyle \frac{{{{K_i}}!}}{{{{k_i}}!( {{{K_i}} - {{k_i}}} )!}}$ is the binomial coefficient, and $\displaystyle{{\mathbb P}_{i,0}[{\theta _1},{\theta _2}, \cdots , \cdots {\theta _{{{k_i}}}}| r_{i,0} ]}$ is the probability that all of $4\times {k_i}$ (i.e., a preamble consists of four preamble symbol groups) preamble symbol groups are successfully transmitted.

	As the  BSs and IoT devices are static all time once they are deployed, the locations of active IoT devices are slightly correlated across time.
	However, the random preamble selection as shown in Fig. 2 randomizes the set of interfering devices over different TTIs , which decorrelates the interference across time, and thus we approximate the distributions of active IoT devices are independent in each TTI following\cite{nan2018random}.
	We ignore the time correlation between each repetition in each TTI due to that the duration of the repetition (6.4 ms) is long enough\cite{name2015}\cite{schlienz2016narrowband}, but we consider the time correlation between the four continuous symbol groups in each repetition.
	
	According to the approximation of the  density of the IoT devices in each CE group in \textbf{Lemma 1}, the Probability Mass Function (PMF) of the number of {intra-group} interfering IoT devices in CE group $i$ in the same cell, i.e., part \Rmnum{1} in  \eqref{RACH} is represented as\cite[Eq.(3)]{yu2013downlink}
	\begin{align}\label{N=n}
	\mathbb{P}[N_i = n_i] = \displaystyle\frac{{{c^{(c + 1)}}\Gamma (n_i + c + 1){{\big( {\displaystyle{{{\mathcal{A}_i^1\mathcal{R}_i^1\lambda _{i}^a}}}/{{{\lambda _B}}}} \big)}^{n_i}}}}{{\Gamma (c + 1)\Gamma (n_i + 1){{\big( {\displaystyle{{{\mathcal{A}_i^1\mathcal{R}_i^1\lambda _{i}^a}}}/{{{\lambda _B}}} + c} \big)}^{n_i + c + 1}}}},
	\end{align}
	where $\lambda_{i}^a$ is given in  \eqref{density}, $c = 3.575$ is a constant related to the approximate PMF of the PPP Voronoi cell, $\Gamma (\cdot)$ is the gamma function, and $\mathcal{A}_i^1\mathcal{R}_i^1$ is the active probability of each IoT device in CE group $i$ in the 1st time slot, where $\mathcal{A}_i^1$ is the non-empty probability (i.e., IoT device buffer is non-empty) and $\mathcal{R}_i^1$ is the non-restrict probability (i.e., IoT device does not defer its access attempt due to RACH scheme).
	
	It is noted that in the 1st time slot, the queue status (number of accumulated packets) of each IoT device only depends on the new packets arrival process ${\Lambda _{New}}$, so we have
	\begin{align}
	\mathcal{A}_i^1= \mathbb{P}\{ N_{New}^1 > 0\}  = 1 - {e^{ - \mu _{New}^1}},   
	\end{align}
	where $\mu_{New}^1$ is the intensity of new arrival packets. 
	Note that the non-restrict probability $\mathcal{R}_i^1$ in the 1st time slot is determined by transmission policies for different RACH schemes, which will be detailed in Section IV..

	In order to solve the RACH success probability of  a randomly chosen IoT device in each CE group, we focus on analyzing the preamble transmission success probability presenting in  \eqref{SI} for three CE groups in the following subsections.
	\subsection{CE Group 0, ${K_0}\le$ 2 (i = 0)}
	The SINR received at the typical BS can be written as
	{
		\begin{align}\label{SINR1}
		{\rm{SINR}} =\displaystyle\frac{\rho h_0}{{\mathcal {I}_0^{{\mathop{ intra}} } +\mathcal {I}_0^{{\mathop{ inter}} }+ {\sigma ^2}}} = \displaystyle\frac{\rho h_0}{{\mathcal {I}_0 + {\sigma ^2}}},
		\end{align}}
	where $\sigma ^2$ is the noise power at the BS, $\mathcal {I}_0$ is the aggregate  interference of the randomly chosen IoT device in CE group 0  and is given as 
	{
		\begin{align}\label{INTRA0}
		\mathcal {I}_{0} = \sum\limits_{j \in {\mathcal {Z}_0^{}}}^{} {P_{0,j}{h_{0,j}}{(r_{0,j})}^{ - \alpha }}. 
		\end{align}}
	In  \eqref{INTRA0},  {$\mathcal Z_0$} is the set of  interfering IoT devices for the typical IoT device in CE group 0, and $h_{0,j}$ is the channel power gain from the interfering IoT device in CE group 0 to the typical BS.
	
	For ease of presentation, we set $l_0 = 4\times k_0$, and the probability that all of $l_0$  preamble symbol groups of the typical IoT device in CE group 0 are successfully transmitted is presented in the following \textbf{Lemma 2.} 
	{
		\begin{lemma}
			The probability that all of $l_0$ received {\rm{SINRs}} at the BS from a {randomly chosen} IoT device  in CE group 0 exceed a certain threshold $\gamma_{th}$ is expressed as
			\begin{align}\label{four}
			&{\mathbb P}_{0,0}[{\theta _1},{\theta _2}, \cdots ,{\theta _{{k_0}}}] \\ \nonumber
			=& \exp \Big( - \frac{{{l_0}{\gamma _{th}}{\sigma ^2} }}{\rho}\Big){\mathbb{E}}\Big[ {\exp \Big( -\frac{{{\gamma _{th}} }}{\rho}\sum\limits_{\beta  = 1}^{l_0} {{I_{0}^\beta}\Big) \big.} } \Big],
			\end{align}			
\setcounter{equation}{17}			
\begin{figure*}[ht]
	\begin{align}\label{class0RACH}
	&{{\cal P}_0^1} =\nonumber\\
	& \sum\limits_{{n_0} = 0}^\infty 
	\Bigg\{
	\underbrace {\frac{{{c^{(c + 1)}}\Gamma ({n_0} + c + 1){{\Big( {\displaystyle{{{\mathcal{A}_0^1\mathcal{R}_0^1}}}{{{\frac{\lambda _{0}^a}{\lambda _B}}}}} \Big)}^{{n_0}}}}}{{\Gamma (c + 1)\Gamma ({n_0} + 1){{\Big( {\displaystyle{{{\mathcal{A}_0^1\mathcal{R}_0^1}}}{{{\frac{\lambda _{0}^a}{\lambda _B}}}} + c} \Big)}^{{n_0} + c + 1}}}}}_{\rm I}
	\underbrace {\sum\limits_{{k_0} = 1}^{{K_0}} {{{( - 1)}^{{k_0} + 1}}\Big( \begin{array}{l}
			{K_0}\\
			{k_0}
			\end{array}\Big)
			\exp \Big(  \frac{{{-l_0}{\gamma _{th}}{\sigma ^2} }}{\rho}
			-\frac{{{2({{\gamma_{th}}}{})^{\frac{2}{\alpha}}\mathcal{A}_0^1\mathcal{R}_0^1\lambda _{0}^a}\gamma \Big( {2,\pi {\lambda _B}{{( {\frac{P}{\rho }} )}^{\frac{2}{\alpha }}}} \Big)}}{{{{ {\lambda _B}}}\Big( {1 - \exp \big( - \pi {\lambda _B}{{( {\frac{P}{\rho }} )}^{\frac{2}{\alpha }}}\big)} \Big)}}\mathcal{F}_0
			\bigg)} }_{{\rm I}{\rm I}}\nonumber\\
	&\underbrace {{{\bigg(1- \sum\limits_{{k_0} = 1}^{{K_0}} {{{( - 1)}^{{k_0} + 1}}\Big( \begin{array}{l}
					{K_0}\\
					{k_0}
					\end{array}\Big)
					\exp \Big(  \frac{{{-l_0}{\gamma _{th}}{\sigma ^2} }}{\rho}
					-\frac{{{2({{\gamma_{th}}}{})^{\frac{2}{\alpha}}\mathcal{A}_0^1\mathcal{R}_0^1\lambda _{0}^a}\gamma \Big( {2,\pi {\lambda _B}{{( {\frac{P}{\rho }} )}^{\frac{2}{\alpha }}}} \Big)}}{{{{ {\lambda _B}}}\Big( {1 - \exp \big( - \pi {\lambda _B}{{( {\frac{P}{\rho }} )}^{\frac{2}{\alpha }}}\big)} \Big)}}\mathcal{F}_0
					\bigg)^{{n_0}}} }}}_{{\rm I}{\rm I}{\rm I}}
	\Bigg\}.
	\end{align}
	\hrulefill	
\end{figure*}	
where the Laplace transform of the aggregate interference  received at the typical BS  is given as 
\setcounter{equation}{15}		
				\begin{align}\label{laplace}
				&{\mathbb E}\Big[\exp \Big( - \frac{{{\gamma _{th}}}}{\rho }\sum\limits_{\beta  = 1}^{{l_0}} {I_{0}^{\beta}\Big) } \Big]\\ \nonumber			
				=& \exp \bigg( - 
				\frac{{{2({{\gamma_{th}}}{})^{\frac{2}{\alpha}}\mathcal{A}_0^1\mathcal{R}_0^1\lambda _{0}^a}\gamma \Big( {2,\pi {\lambda _B}{{( {\frac{P}{\rho }} )}^{\frac{2}{\alpha }}}} \Big)}}{{{{ {\lambda _B}}}\Big( {1 - \exp \big( - \pi {\lambda _B}{{( {\frac{P}{\rho }} )}^{\frac{2}{\alpha }}}\big)} \Big)}}\mathcal{F}_0\bigg),
				\end{align}					  	
			where 
			\begin{align}\label{F0}
			\mathcal{F}_0=\displaystyle\int_{({\gamma_{th}})^{\frac{-1}{\alpha}}}^{{\infty}}{\Big[ {1 - {{\big( {\frac{1}{{1 + y^{ - \alpha }}}} \big)}^{{l_0}}}} \Big]} ydy.  
			\end{align}
	\end{lemma}}
	\begin{proof}
		See Appendix B.
	\end{proof}
	Substituting   \eqref{four} into \eqref{SI}, we obtain the preamble transmission success probability 
	and then substituting  \eqref{N=n} and \eqref{SI} into \eqref{RACH}, we derive the RACH  success probability of a randomly chosen IoT device in CE group 0 in the 1st time slot in the following \textbf{Theorem 1}.
	
	\begin{theorem}
		The RACH  success probability of a randomly chosen IoT device in the CE group 0 in the 1st time slot is derived in \eqref{class0RACH} at the top of this page with $\mathcal{F}_0$  given in \eqref{F0}.
	\end{theorem}
	
	\subsection{{CE group 1 and CE group 2 in Case 1}, ${K_i}>2$ (i = 1, 2)}
	The SINR received at the typical BS is written as
	\setcounter{equation}{18}	
	\begin{align}\label{SINR2}
	{\rm{SINR}} =\displaystyle \frac{{Ph_{i,0}(r_{i,0})^{ - \alpha }}}{{\mathcal {I}_{i} + {\sigma ^2}}},
	\end{align}
	where $\mathcal {I}_{i}$ is aggregate interference of the randomly chosen IoT device in CE group $i$ given as 
	\begin{align}\label{INTRA2}
	\mathcal {I}_{i} = \sum\limits_{j \in {\mathcal {Z}_i}}^{} {P{h_{i,j}}({r_{i,j}})^{ - \alpha }}. 
	\end{align}
	In  \eqref{INTRA2},  $\mathcal Z_i$ is the set of interfering IoT devices for the typical IoT device in CE group $i$, $h_{i,j}$ and $r_{i,j}$ are channel power gain and distance from the  interfering IoT device to the typical BS. 
	
	According to the nature of the Poisson Process, given that there are $N_i+1$ IoT devices in the area of $CA_i$, ${r_{i,0}}$ follows independent and identical uniform distribution\cite{56383}. 
	Let $R$ denote the random variable with the same uniform distribution, the PDF of $R$ is derived as 
	\begin{align}\label{distance}
	{f_{{R}}}({r}) = 2{r}/(D_{i}^2 - D_{{i - 1}}^2), ({D_{i-1}} \le {r} \le {D_i}),
	\end{align}
	where $D_0$ , $D_1$ are given in  \eqref{T_i_2} and $D_2$ is given in  \eqref{D2}.

	{Same as CE group 0, we set $l_i = 4\times {k_i}$, and then we  derive the probability that all of $l_i$ preamble symbol groups of the typical IoT device in CE group $i$ $(i = 1, 2)$  are successfully transmitted in the following \textbf{Lemma 3}.}
	{
		\begin{lemma}\label{preanbleii}
			The probability that all of $l_i$ received {\rm{SINRs}} at the BS from a {randomly chosen} IoT device  in CE group $i$  $(i = 1, 2)$ exceed a certain threshold $\gamma_{th}$ is expressed as
			\begin{align}\label{preamblei}
			&p({\gamma _{th}})\\ \nonumber
			&= \int\limits_{{D_{i - 1}}}^{{D_i}} {\exp\Big( -  {{\frac{{{l_i\gamma _{th}\sigma ^2}r^\alpha }}{P}}}\Big)}
			\exp \big( - 2\pi {\mathcal{A}_i^1\mathcal{R}_i^1\lambda _{i}^a}\mathcal{F}_i \big)f_R(r)dr
			\end{align} 
			where 
			\begin{align}\label{FI}
			\mathcal{F}_i=\displaystyle\int\limits_{{D_{i}}}^{{\infty}} {\Big( {1 - {{\Big( {\frac{1}{{1 + \gamma _{th}r^\alpha{{y}^{ - \alpha }}}}} \Big)}^{{l_i}}}} \Big)} ydy.  
			\end{align}
			and $f_R(r)$ is given in \eqref{distance}.
	\end{lemma}}
	
	\begin{proof}
		See Appendix C.
	\end{proof}
	
	{Note that the above \textbf{Lemma 3.} is suitable for CE group 2 in Case 1.
		For CE group 2 in Case 2, we have
		the following \textbf{Lemma 4.}
		\begin{lemma}\label{case2}
			The probability that all of $l_2$ received {\rm{SINRs}} at the BS from a {randomly chosen} IoT device  in CE group $2$   exceed a certain threshold $\gamma_{th}$ is expressed as
			\begin{align}\label{preamble2}
			&p({\gamma _{th}})\\ \nonumber
			&= \int\limits_{{D_1}}^{{\infty}} {\exp\Big( -  {{\frac{{{l_2\gamma _{th}\sigma ^2}r^\alpha }}{P}}}\Big)}
			\exp \big( - 2\pi {\mathcal{A}_2^1\mathcal{R}_2^1\lambda _{2}^a}\mathcal{F}_2 \big)f_R(r)dr
			\end{align} 
			where 
			\begin{align}\label{FI2}
			\mathcal{F}_2=\displaystyle\int\limits_{{D_{1}}}^{{\infty}} {\Big( {1 - {{\Big( {\frac{1}{{1 + \gamma _{th}r^\alpha{{y}^{ - \alpha }}}}} \Big)}^{{l_2}}}} \Big)} ydy.  
			\end{align}
			and 
			\begin{align}\label{r_0case2}
			{f_{{R}}}({r}) = 2\pi{\lambda_B}r\exp(-{\lambda_B}{\pi}(r^2-D_1^2)).
			\end{align}
	\end{lemma}}

		\begin{figure*}
		\begin{align}\label{class2RACH}
		&{{\cal P}_i^1} =\nonumber\\
		&		 \sum\limits_{{n_i} = 0}^\infty 
		\Bigg\{ 
		\underbrace {\displaystyle\frac{{{c^{(c + 1)}}\Gamma ({n_i} + c + 1){{\Big( {\displaystyle{{{\mathcal{A}_i^1\mathcal{R}_i^1}}}{{{\frac{\lambda _{i}^a}{\lambda _B}}}}} \Big)}^{{n_i}}}}}{{\Gamma (c + 1)\Gamma ({n_i} + 1){{\Big( {\displaystyle{{{\mathcal{A}_i^1\mathcal{R}_i^1}}}{{{\frac{\lambda _{i}^a}{\lambda _B}}}} + c} \Big)}^{{n_i} + c + 1}}}}}_{\rm I}\underbrace {\sum\limits_{{{k_i}} = 1}^{{{K_i}}} {{{( - 1)}^{{{k_i}} + 1}}\Big( \begin{array}{l}
				{{K_i}}\\
				{{k_i}}
				\end{array} \Big)
				\int\limits_{{D_{i - 1}}}^{{D_i}} {\exp\Big(   {{\frac{{{-l_i\gamma _{th}\sigma ^2}r^\alpha }}{P}}}}
				 - 2\pi {\mathcal{A}_i^1\mathcal{R}_i^1\lambda _{i}^a}\mathcal{F}_i \Big)
				{f_{{R}}}({r})d{r} 
		} }_{{\rm I}{\rm I}}\nonumber\\
		&\underbrace {{{\bigg( {1 - \sum\limits_{{{k_i}} = 1}^{{{K_i}}} {{{( - 1)}^{{{k_i}} + 1}}\Big( \begin{array}{l}
							{{K_i}}\\
							{{k_i}}
							\end{array} \Big)
							\int\limits_{{D_{i - 1}}}^{{D_i}} {\exp\Big(  {{\frac{{{-l_i\gamma _{th}\sigma ^2}r^\alpha }}{P}}}}
							 - 2\pi {\mathcal{A}_i^1\mathcal{R}_i^1\lambda _{i}^a}\mathcal{F}_i \Big)
							{f_{{R}}}({r})d{r} 
						} 
					} \bigg)}^{{n_i}}}}_{{\rm I}{\rm I}{\rm I}}
		\Bigg\},
		\end{align}
		\hrulefill
	\end{figure*}
	Substituting   \eqref{preamblei} or \eqref{preamble2}  into \eqref{SI}, we obtain the preamble transmission success probability 
	and then substituting  \eqref{N=n} and \eqref{SI} into \eqref{RACH}, we derive the RACH  success probability of a randomly chosen IoT device in CE group $i$ $(i = 1, 2)$  in the 1st time slot in the following \textbf{Theorem 2}.
	{
		\begin{theorem}
			The RACH  success probability of a randomly chosen IoT device in the CE group $i$ in the 1st time slot is derived in \eqref{class2RACH} at the top of this page, where $\mathcal{F}_i$  is given in \eqref{FI} and $f_R(r)$ is given in \eqref{distance} for CE group 1 and CE group 2 in Case 1; $\mathcal{F}_i$  is given in \eqref{FI2} and $f_R(r)$ is given in \eqref{r_0case2} for  CE group 2 in Case 2.
	\end{theorem}}


	

	\section{Multiple Time Slots Model}
	This section focuses on the RACH success probability of the IoT device in CE group $i$ in NB-IoT network over multiple time slots with different RACH schemes.
	Apart from the physical layer modeling in the spatial domain based on stochastic geometry, the queue evolation in the time domain is modeled and analyzed using probability theory.
	Note that inactive IoT devices 
	do not attempt RACH, such that they do not generate interference. 
	As mentioned before, whether an IoT device is active or not in the $t$th TTI depends on the non-empty probability $\mathcal{A}_i^t$ and the non-restrict probability $\mathcal{R}_i^t$ of each IoT device.
	Mathematically, to derive the RACH success probability ${\cal P}_i^t$ of a randomly chosen IoT device in CE group $i$ in the $t$th time slot, we need to derive the non-empty probability $\mathcal{A}_i^t$ and the non-restrict probability $\mathcal{R}_i^t$ of the IoT device, which are decided by ${\cal P}_i^{t-1}$, $\mathcal{A}_i^{t-1}$, and $\mathcal{R}_i^{t-1}$.

	Following our precious work \cite{nan2018random}, the accumulated packets number $N_{Cum,i}^t$ of an IoT device for CE group $i$ in the $t$th time slot could be approximated as Poisson distribution $\Lambda_{Cum,i}^t$ with intensity $\mu_{Cum,i}^t$.
	Then the non-empty probabilities $\mathcal{A}_i^{t}$ $(t>1)$ of each IoT device for CE group $i$ in the $t$th time slot are derived based on the iteration process below.
	{
		\begin{align}\label{non-empty probability}
		\begin{cases}
		{\mathcal{A}_i^t} =\mathbb P\{ N_{New}^t + N_{Cum,i}^t > 0\}= 1 - {e^{ - \mu _{New}^t - \mu _{Cum,i}^t}},\\
		\mu _{Cum,i}^t = \mu _{New}^{t - 1} + \mu _{Cum,i}^{t - 1} -g_i{\mathcal{P}_i^{t - 1}}{\mathcal{R}_i^{t - 1}}{\mathcal{A}_i^{t - 1}} .
		\end{cases}
		\end{align}}
	
	In order to derive the RACH success probability $\mathcal{P}_i^{t}$ of a randomly chosen IoT device in CE group $i$ in the $t$th time slot, we also need to have the non-restrict probability $\mathcal{R}_i^{t}$. Note that for different RACH schemes, $\mathcal{R}_i^{t}$ are determined by their transmission policies. 
	\subsubsection{Baseline Scheme} 
	The baseline scheme allows each IoT device to attempt RACH immediately when there are packets in the buffer, so the non-restrict probability in any time slot is given as
	\begin{align}\label{BL}
	\mathcal{R}_{BL}^t = 1.
	\end{align}
	Substituting  \eqref{BL} into  \eqref{non-empty probability}, we have
	\begin{align}\label{BL non-empty probability}
	\begin{cases}
	{\mathcal{A}_{i,BL}^t} = 1 - {\displaystyle e^{ - \mu _{New}^t - \mu _{Cum,i,BL}^t}},\\
	\mu _{Cum,i,BL}^t = \mu _{New}^{t - 1} + \mu _{Cum,i,BL}^{t - 1} - {g_i}{\mathcal{P}_{i,BL}^{t - 1}}{\mathcal{A}_{i,BL}^{t - 1}}.
	\end{cases}
	\end{align}

	\subsubsection{ACB Scheme} 
	The BS initially broadcasts an ACB factor ${\rm Q}_{ACB}$, and then an non-empty IoT device draws a random number $q \in [0,1]$, and compares this with ${\rm Q}_{ACB}$. 
	Each non-empty IoT device is allowed to perform RACH procedure only if $q < {\rm Q}_{ACB}$. So we have the non-restrict probability in any time slot as
	\begin{align}\label{ACB}
	\mathcal{R}_{ACB}^t = Q_{ACB}.
	\end{align}
	Substituting  \eqref{ACB} into  \eqref{non-empty probability}, we have
	\begin{align}\label{ACB non-empty probability}
	\begin{cases}
	{\mathcal{A}_{i,ACB}^t} = 1 - {e^{ - \mu _{New}^t - \mu _{Cum,i,ACB}^t}},\\
	\mu _{Cum,i,ACB}^t = \mu _{New}^{t - 1} + \mu _{Cum,i,ACB}^{t - 1} - {g_i}{\rm{Q}_{ACB}}{\mathcal{P}_{i,ACB}^{t - 1}}{\mathcal{A}_{i,ACB}^{t - 1}}.
	\end{cases}
	\end{align}
	
	\subsubsection{BO Scheme}
	The analysis of the BO scheme is similar to the ACB scheme, due to the BO procedure can be visualised as a group of IoT devices are completely barred for a time slot. 
	In the 1st time slot, none of IoT device defers the attempt, such that the transmission procedure is the same as the baseline scheme. 
	After the 1st time slot, if a RACH attempt fails, the BO mechanism is executed, where the non-empty IoT devices defer their RACH attempts and wait for $T_{BO}$ time slots. 
	Due to the BO mechanism, only non-empty IoT devices without RACH attempt failures in the last $T_{BO}$ time slots can attempt RACH, and only those IoT devices generate interference that affect the RACH success probability in the $t$th time slot.
	The non-restrict probability (i.e., the probability of non-empty IoT devices in CE group $i$ do not defer their RACH attempt) $\mathcal{R}_{i,BO}^t$ is
	\begin{align}\label{BACKOFF}
&	\mathcal{R}_{i,BO}^t = \\ \nonumber
&	\left\{ \begin{array}{l}
	1 - \displaystyle\frac{{\sum\limits_{s =1}^{t - 1} {(1 - {g_i}{\mathcal{P}_{i,BO}^{t-s}}){\mathcal{A}_{i,BO}^{t-s}}{\mathcal{R}_{i,BO}^{t-s}}} }}{{{\mathcal{A}_{i,BO}^t}}}, t \le {T_{BO}} + 1,\\
	1 - \displaystyle\frac{{\sum\limits_{s = 1}^{{T_{BO}}} {(1 - {g_i}{\mathcal{P}_{i,BO}^{t-s}}){\mathcal{A}_{i,BO}^{t-s}}{\mathcal{R}_{i,BO}^{t-s}}} }}{{{\mathcal{A}_{i,BO}^t}}},t > {T_{BO}}.
	\end{array} \right.
	\end{align}
	Substituting  \eqref{BACKOFF} into  \eqref{non-empty probability}, we have
	\begin{align}\label{BO non-empty probability}
	\begin{cases}
	{\mathcal{A}_{i,BO}^t} = 1 - \displaystyle{e^{ - \mu _{New}^t - \mu _{Cum,i,BO}^t}},\\
	\mu _{Cum,i,BO}^t = \mu _{New}^{t - 1} + \mu _{Cum,i,BO}^{t - 1} - {g_i}{\mathcal{A}_{i,BO}^{t - 1}}{\mathcal{P}_{i,BO}^{t-1}}{\mathcal{R}_{i,BO}^{t-1}}.
	\end{cases}
	\end{align}
	
	\subsubsection{ACB\&BO Scheme}
	The ACB\&BO scheme is an integrated scheme with 
	combined the ACB and BO schemes, so the the non-restrict probability can be derived following  \eqref{BACKOFF} as
	{
		\begin{align}\label{ACBBO}
		&\mathcal{R}_{i,ACB\&BO}^t = \\ \nonumber
	&	\left\{ \begin{array}{l}
		1 - \displaystyle\frac{\displaystyle{{\sum\limits_{s =1}^{t - 1} {(1 - {g_i}{\rm{Q}}_{ACB}{\mathcal{P}_{i,ACB\&BO}^{t-s}}){\mathcal{A}_{i,ACB\&BO}^{t-s}}{\mathcal{R}_{i,ACB\&BO}^{t-s}}} }}{}{}}{{\mathcal{A}_{i,ACB\&BO}^t}},
		\\t \le {T_{BO}} + 1,\\
		1 - \displaystyle\frac{\displaystyle{{\sum\limits_{s = 1}^{{T_{BO}}} {(1 - {g_i}{\rm{Q}}_{ACB}{\mathcal{P}_{i,ACB\&BO}^{t-s}}){\mathcal{A}_{i,ACB\&BO}^{t-s}}{\mathcal{R}_{i,ACB\&BO}^{t-s}}} }}{}}{{{\mathcal{A}_{i,ACB\&BO}^t}}},\\t > {T_{BO}}.
		\end{array} \right.
		\end{align}}
	
	Substituting  \eqref{ACBBO} into  \eqref{non-empty probability}, we have
	\begin{align}\label{ACBBO non-empty probability}
	\begin{cases}
	{\mathcal{A}_{i,ACB\&BO}^t} = 1 - \displaystyle{e^{ - \mu _{New}^t - \mu _{Cum,i,ACB\&BO}^t}},\\
	\mu _{Cum,i,ACB\&BO}^t = \mu _{New}^{t - 1} + \mu _{Cum,i,ACB\&BO}^{t - 1} \\ 
	 -{g_i}{\rm{Q}}_{ACB}{\mathcal{P}_{i,ACB\&BO}^{t - 1}}{\mathcal{A}_{i,ACB\&BO}^{t - 1}}{\mathcal{R}_{i,ACB\&BO}^{t-1}}.
	\end{cases}
	\end{align}

	The RACH success probability of a randomly chosen IoT device in each CE group in the $t$th time slot for all RACH schemes is presented in the following Theorem.
	\begin{theorem}
		The RACH success probability of a randomly chosen IoT device in each CE group in the $t$th time slot for all RACH schemes is derived as
		{
			\begin{align}\label{mslot}
			\mathcal{P}_i^t = \sum\limits_{{n_i} = 0}^\infty  {\bigg\{ {{\rm O}[{n_i},t]\Theta [{{K_i}},t]{{\Big( {1 - \Theta [{{K_i}},t]} \Big)}^{{n_i}}}} \bigg\}}.
			\end{align}}
		In  \eqref{mslot}, the probability of the number of {intra-group}  interfering IoT devices is derived as
		\begin{align}
		{\rm O}[{n_i},t]=\frac{{{c^{(c + 1)}}\Gamma (n_i + c + 1){{\Big( {\displaystyle\frac{{{\mathcal{A}_i^t}{\mathcal{R}_i^t}{\lambda _i^a}}}{{{\lambda _B}}}} \Big)}^{n_i}}}}{{\Gamma (c + 1)\Gamma (n_i + 1){{\Big( {\displaystyle\frac{{{\mathcal{A}_i^t}{\mathcal{R}_i^t}{\lambda _i^a}}}{{{\lambda _B}}} + c} \Big)}^{n_i + c + 1}}}}, 
		\end{align}
		where $\mathcal{A}_i^t$ and $\mathcal{R}_i^t$ are iteratively updated using  \eqref{BL non-empty probability}- \eqref{ACBBO non-empty probability}
		respectively, for Baseline, ACB, BO and ACB\&BO schemes;
		and the preamble transmission success probability with ${K_i}$ repetitions is derived as		
			\begin{align}
			&\Theta [{K_0},t]= \sum\limits_{{k_0} = 1}^{{K_0}} {{( - 1)}^{{k_0} + 1}}\Big( \begin{array}{l}
				{K_0}\\
				{k_0}
				\end{array}\Big)
				\\ \nonumber
			&	\exp \Big( - \frac{{{l_0}{\gamma _{th}}{\sigma ^2} }}{\rho}
				-\frac{{{2({{\gamma_{th}}}{})^{\frac{2}{\alpha}}\mathcal{A}_0^t\mathcal{R}_0^t\lambda _{0}^a}\gamma \Big( {2,\pi {\lambda _B}{{( {\frac{P}{\rho }} )}^{\frac{2}{\alpha }}}} \Big)}}{{{{ {\lambda _B}}}\Big( {1 - \exp \big( - \pi {\lambda _B}{{( {\frac{P}{\rho }} )}^{\frac{2}{\alpha }}}\big)} \Big)}}\mathcal{F_0}
				\bigg) 
			\end{align}
		for CE group 0,
		and 
		{
			\begin{align}
		&	\Theta [{{K_i}},t]=\sum\limits_{{{k_i}} = 1}^{{{K_i}}} {{( - 1)}^{{{k_i}} + 1}}\Big( \begin{array}{l}
				{{K_i}}\\
				{{k_i}}
				\end{array} \Big)
				\\ \nonumber
			&	\int\limits_{{D_{i - 1}}}^{{D_i}} {\exp\Big( -  {{\frac{{{l_i\gamma _{th}\sigma ^2}r^\alpha }}{P}}}\Big)}
				\exp \big( - 2\pi {\mathcal{A}_i^t\mathcal{R}_i^t\lambda _{i}^a}\mathcal{F}_i \big)
				f_R(r)
				d{r} 
			\end{align}}
		for CE group 1 and 2, with $\mathcal{F}_i$ and $f_R(r)$ given in \textbf{Lemma 1-Lemma 3}.
	\end{theorem}
	
	\section{Simulation and Discussion}
	In this section, the derived analytical results are validated via Monte Carlo simulations.
	The system simulation parameters are summarized in Table I following \cite{name2015}.
	\begin{table}[htbp!]
		\centering
		\caption{ Simulation Parameters}
		{\renewcommand{\arraystretch}{1.2}
			\begin{tabular}{|c|c|}
				\hline
				
				NB-IoT Bandwidth & 180 kHz   \\ \hline
				NPRACH Subcarrier Spacing &  3.75 kHz  \\ \hline
				Symbol Group  &  1 CP and 5 symbols \\ \hline
				NPRACH Band (3 CE groups) &  12 ,12 and 24 subcarriers  \\ \hline
				Transmit Power  & {35} dBm (DL), {22} dBm (UL) \\ \hline
				Noise Figure  & 5 dB (DL), 3 dB (UL) \\ \hline
			\end{tabular}
		}
		\label{table_accord}
	\end{table}
	The BSs and IoT devices are deployed via independent PPPs in a {40000} km$^2$ circle area.
	The real buffer at each IoT device is simulated to capture the packets arrival and accumulation process evolved over time.
	Furthermore, in the ACB scheme, we also simulate that each IoT device
	generates a random number $q ∈ [0, 1]$ and compares with the ACB factor $Q_{ACB}$ to determine whether the current RACH is deferred, and in the Back-Off scheme, we capture all RACH failures and practically defer RACH attempts of these IoT devices for the next $T_{BO}$ time slots.
	Unless otherwise stated, we set
	{$\lambda_{\rm{B}}=0.1$ BSs/km$^2$, $\lambda_{\rm{D}}=10$ IoT devices/km$^2$, $\gamma_{th}= 10$ dB, $\alpha = 4$, and $\rho =−120$ dBm. 
		The noise is $\sigma^2$ = −174+5+10log$_{10}$(180000) = −116.4 dBm and $\omega$ = −174+3+10log$_{10}$(3750) = −135.3 dBm.
		The target minimum SNRs for the three CE groups are $\delta_{1}=35$ dB and $\delta_{2}=30$ dB, respectively.
		We choose the same new packets arrival rate for each time slot ($\mu _{New}^1 = \mu _{New}^2 = ... = \mu _{New}^m = 0.1$ packets/time slot).
		Unless otherwise stated, we consider $T_{BO} = 2$ for BO scheme and $Q_{ACB} = 0.6$ for the  ACB scheme.
	}
	\begin{figure}[htbp!]
		\begin{center}
				\centering
				\includegraphics[width=3.0in,height=2.3in]{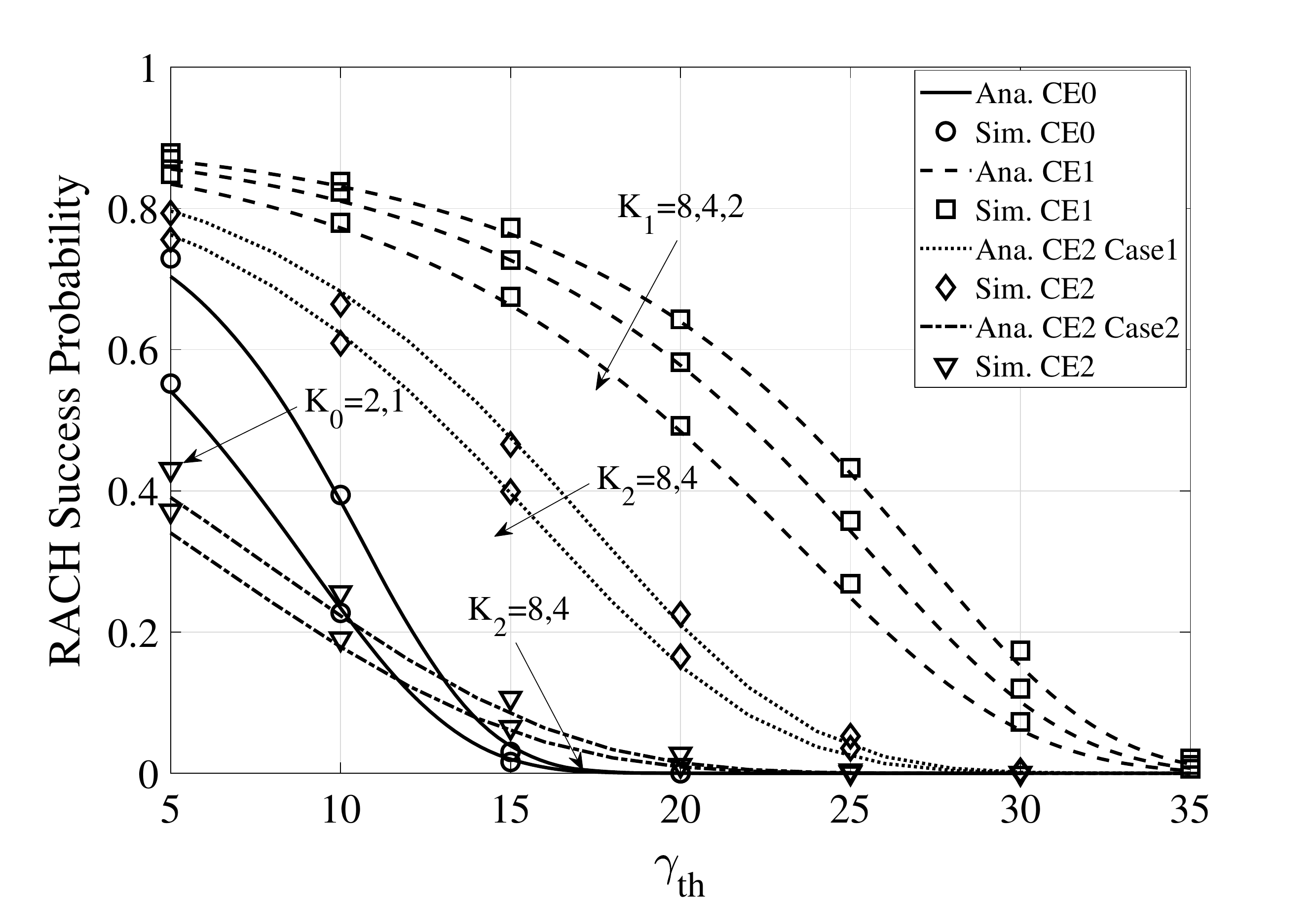}
				\vspace*{-0.1cm}
				\caption{\scriptsize  RACH success probability for three CE groups versus $\gamma_{th}$ in the 1st time slot}
			\label{fig:3}
		\end{center}
	\end{figure}
	{Fig. 3 plots the RACH success probability of a randomly chosen IoT device in the three CE groups in the 1st time slot using  \eqref{class0RACH} and  \eqref{class2RACH} versus the SINR threshold for various repetition values.
		We first observe a good match between the analysis and the simulation results, which validates the accuracy of the developed mathematical framework.
		We observe that the RACH  success probability degrades with the increase of the SINR threshold. 
		According to  \eqref{SI}, increasing $\gamma_{th}$ leads to lower preamble transmission success probability of both interfering IoT devices and serving IoT device, thereby decreasing the overall RACH success probability. 
		We also observe that the RACH success probabilities of IoT devices in CE group 1 and CE group 2 in Case 1 are higher than that in CE group 0, which indicates that increasing the repetition value leads to higher RACH success probability and could ensure the RACH performance with extended coverage.
		In addition, we note that the RACH success probability of the CE group 2 in Case 2 is low as there are a large number of  IoT devices in CE 2 in Case 2, where the external radius of  CE group 2 equals to the Voronoi cell radius.
	}


	
	
		\begin{figure}[htbp!]
		\begin{center}
							\centering
							\includegraphics[width=3.0in,height=2.3in]{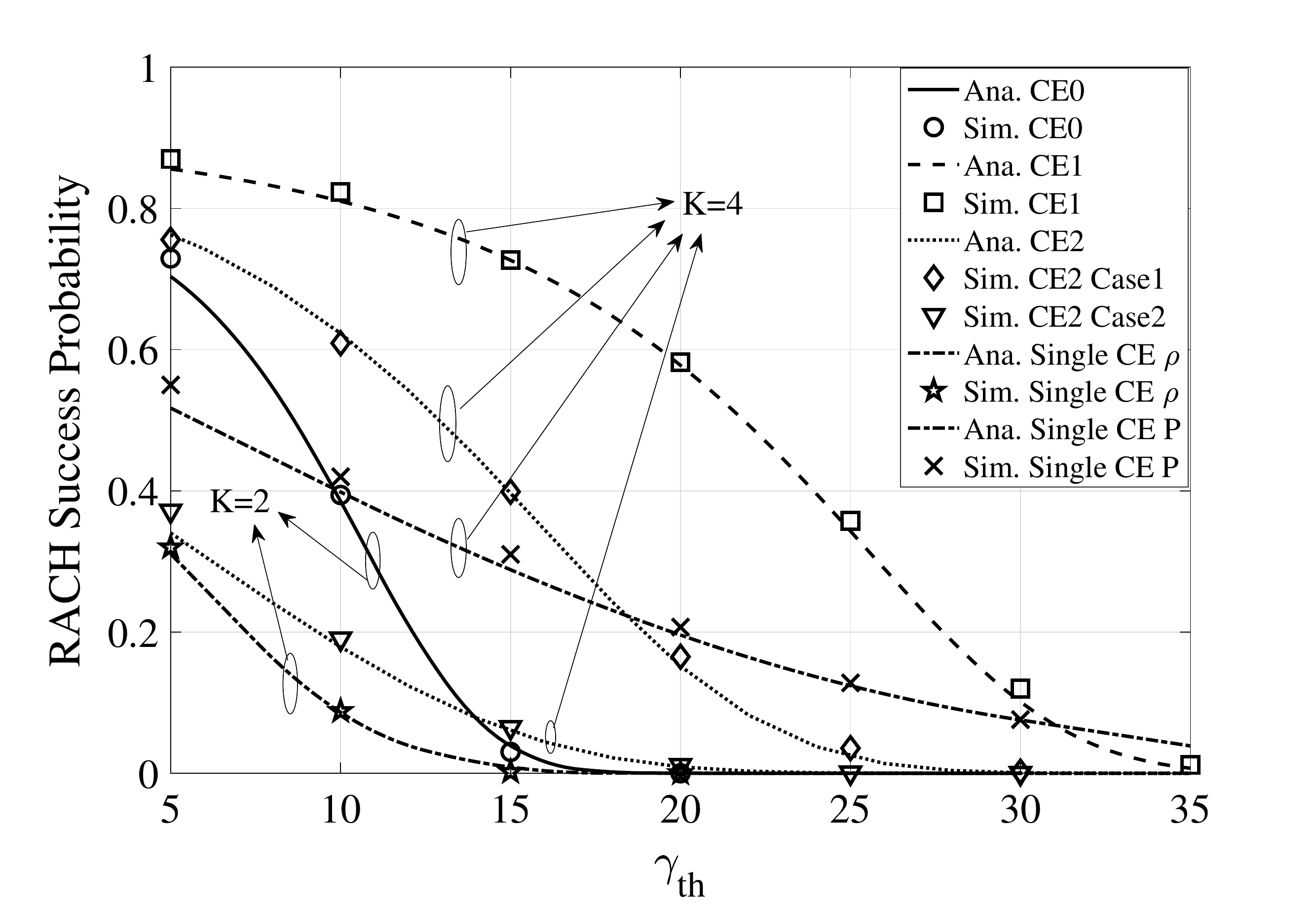}
							\vspace*{-0.1cm}
							\caption{\scriptsize RACH success probability for three CE groups and single CE group}
						\label{fig:4}
		\end{center}
	\end{figure}
	{Fig. 4 compares the RACH success probabilities of the device in an NB-IoT network with three CE groups with that in a single CE group NB-IoT network (using power control threshold $\rho$ and fixed transmit power $P$, respectively).
		It is obvious that the RACH success probabilities of the devices in three CE groups (except CE group 2 in Case 2) greatly outperform that in a single CE group network.
		For example, 1) the RACH success probability of the device in CE group 1 with 4 repetitions is two times more than that in a single CE group with the same repetition value and same transmit power when the SINR threshold $\gamma_{th}\le 25$ dB; 2) the  RACH success probability of the device in CE group 0 with 2 repetitions is two times more than that in a single CE group with the same repetition value and power control when the SINR threshold $\gamma_{th}\le 10$ dB.
		Interestingly, the RACH success probabilities of the devices in CE group 2 in Case 2 are lower than those in the single CE group.
		This is due to that a lot of IoT devices are in CE group 2 but the configured preamble set $S_2=24$ is much smaller than the total number 48 for a single CE group.
		Thus,  categorizing the IoT devices into up to three CE groups is not always beneficial to all the groups,  which is affected by the choice of the categorizing parameters.
	}

	\begin{figure}[htbp!]
		\begin{center}
			\label{fig:5}
			\begin{minipage}[t]{0.48\textwidth}
				\centering
				\includegraphics[width=3.0in,height=2.3in]{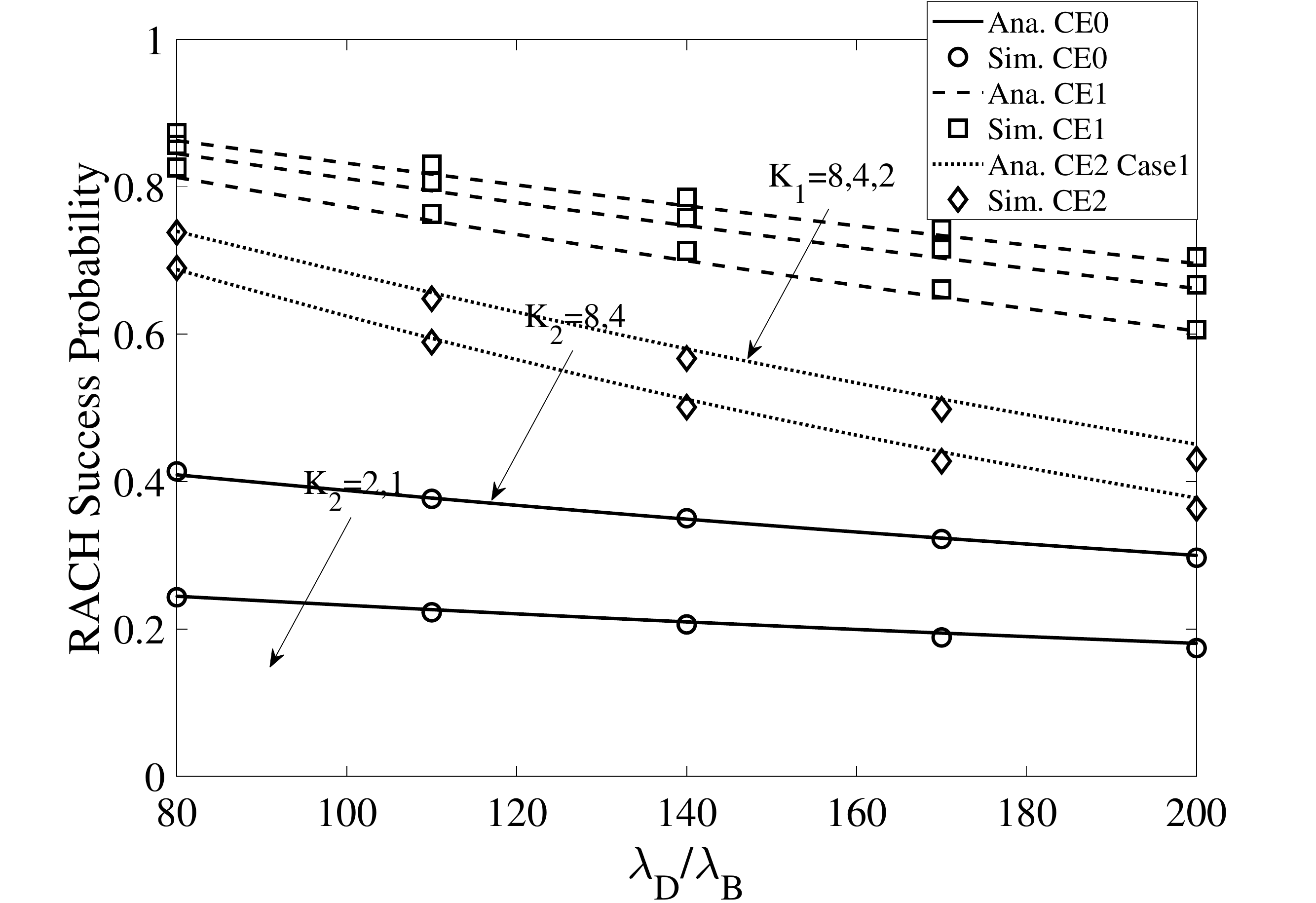}
				\vspace*{-0.1cm}
				\caption{\scriptsize RACH success probability for three CE groups versus
					density ratio $\lambda_D/\lambda_B$ in the 1st time slot}
			\end{minipage}
			\label{fig:6}
		\end{center}
	\end{figure}
	{Fig. 5} plots the RACH success probabilities of a randomly chosen IoT device for three CE groups versus the density ratios $\lambda_D/\lambda_B$ in the 1st time slot for various repetition values ${K_i}$. 
	We first observe that the RACH success probability decreases with the increase of the density ratio between IoT devices and BSs ($\lambda_D/\lambda_B$), which is due to the following two reasons: 1) increasing the number of IoT devices generating interference leads to lower received SINR at the BS; 2) increasing the number of IoT devices leads to a higher probability of collision. 
	\begin{figure*}[htbp!]
		\centering
		\subfigure[CE group 0] {\includegraphics[width=3in,height=2.2in]{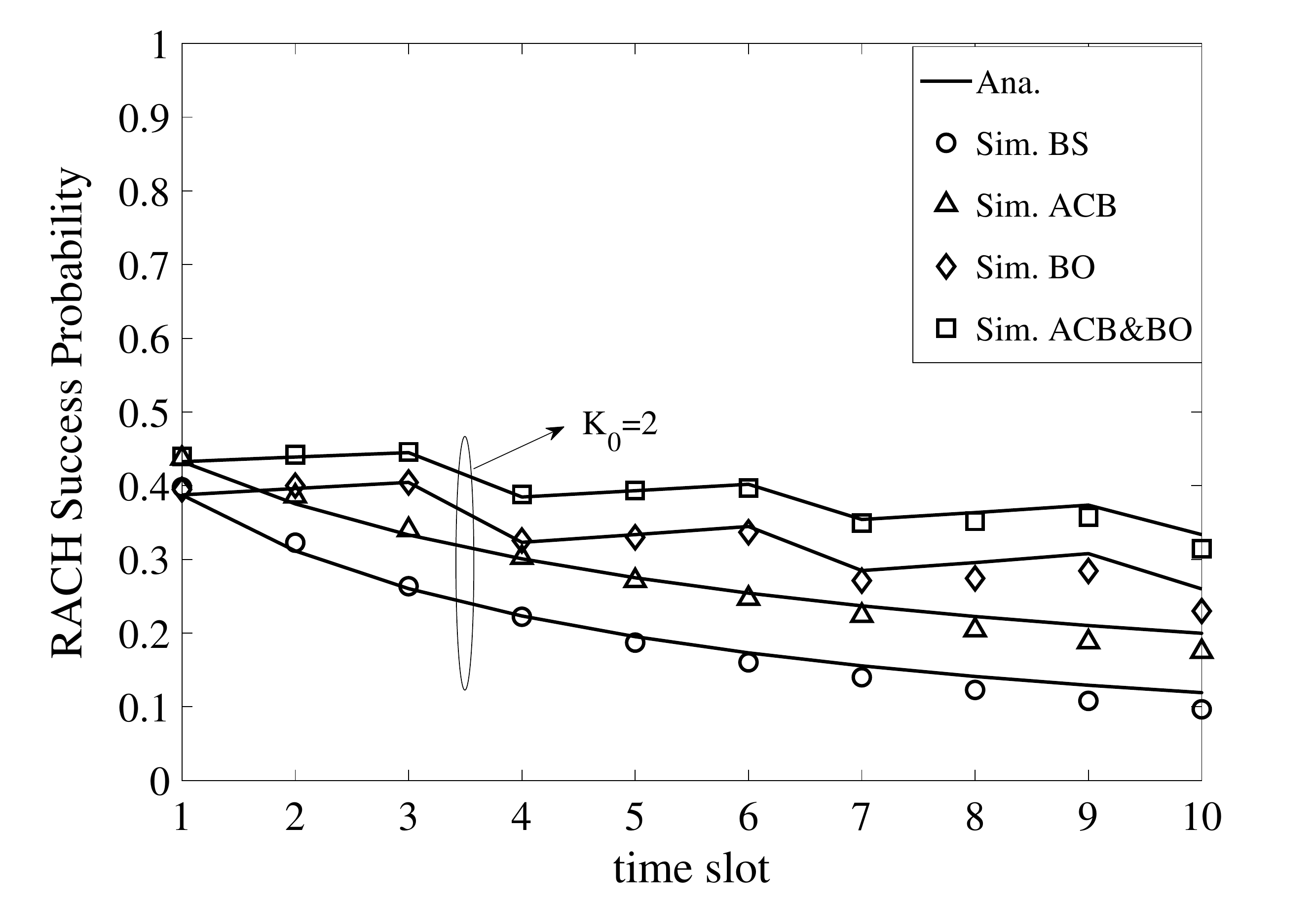}}
		\subfigure[CE group 1] {\includegraphics[width=3in,height=2.2in]{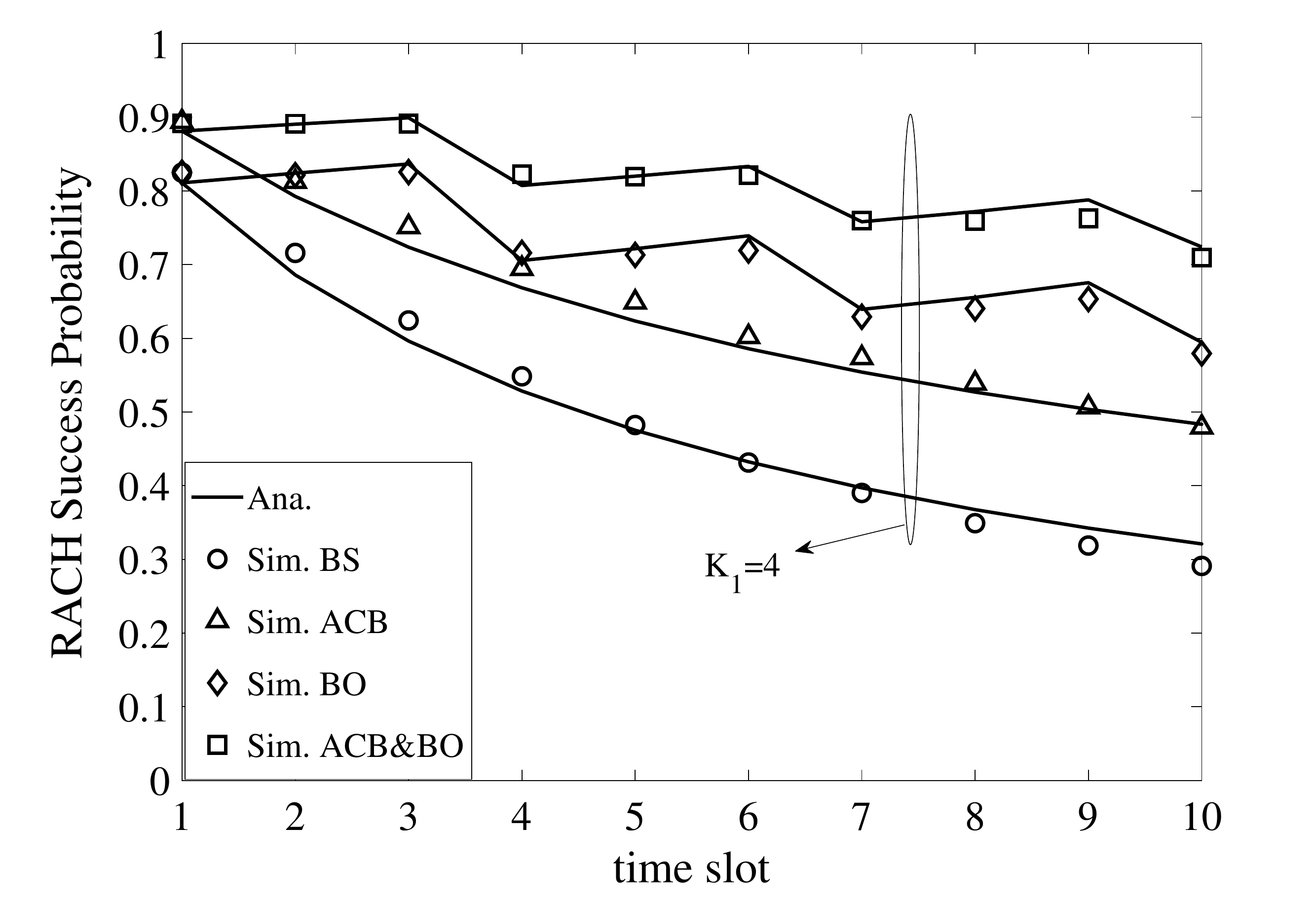}}
		\subfigure[CE group 2 in case 1] {\includegraphics[height=2.2in,width=3in]{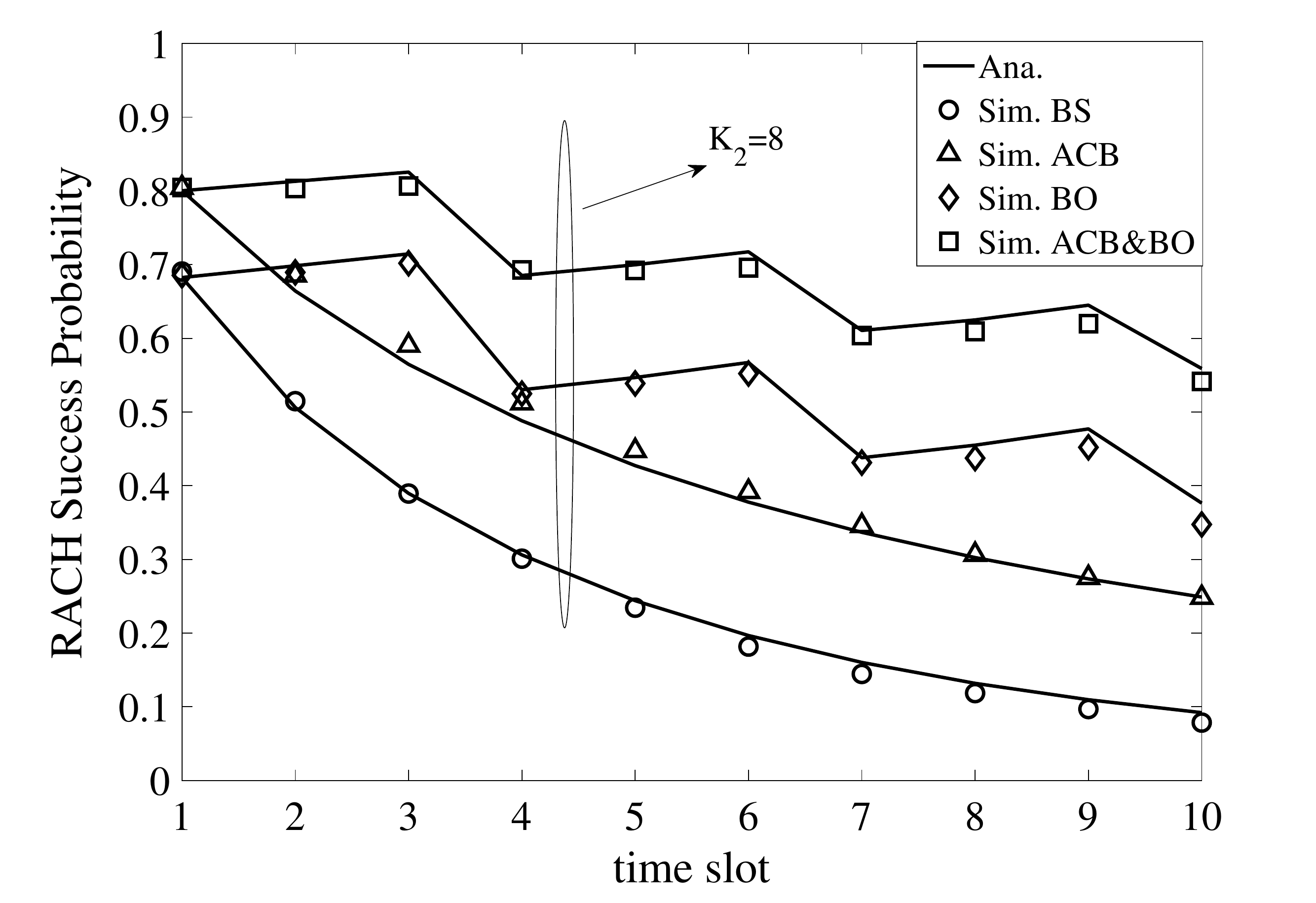}}
		\subfigure[CE group 2 in case 2] {\includegraphics[height=2.2in,width=3in]{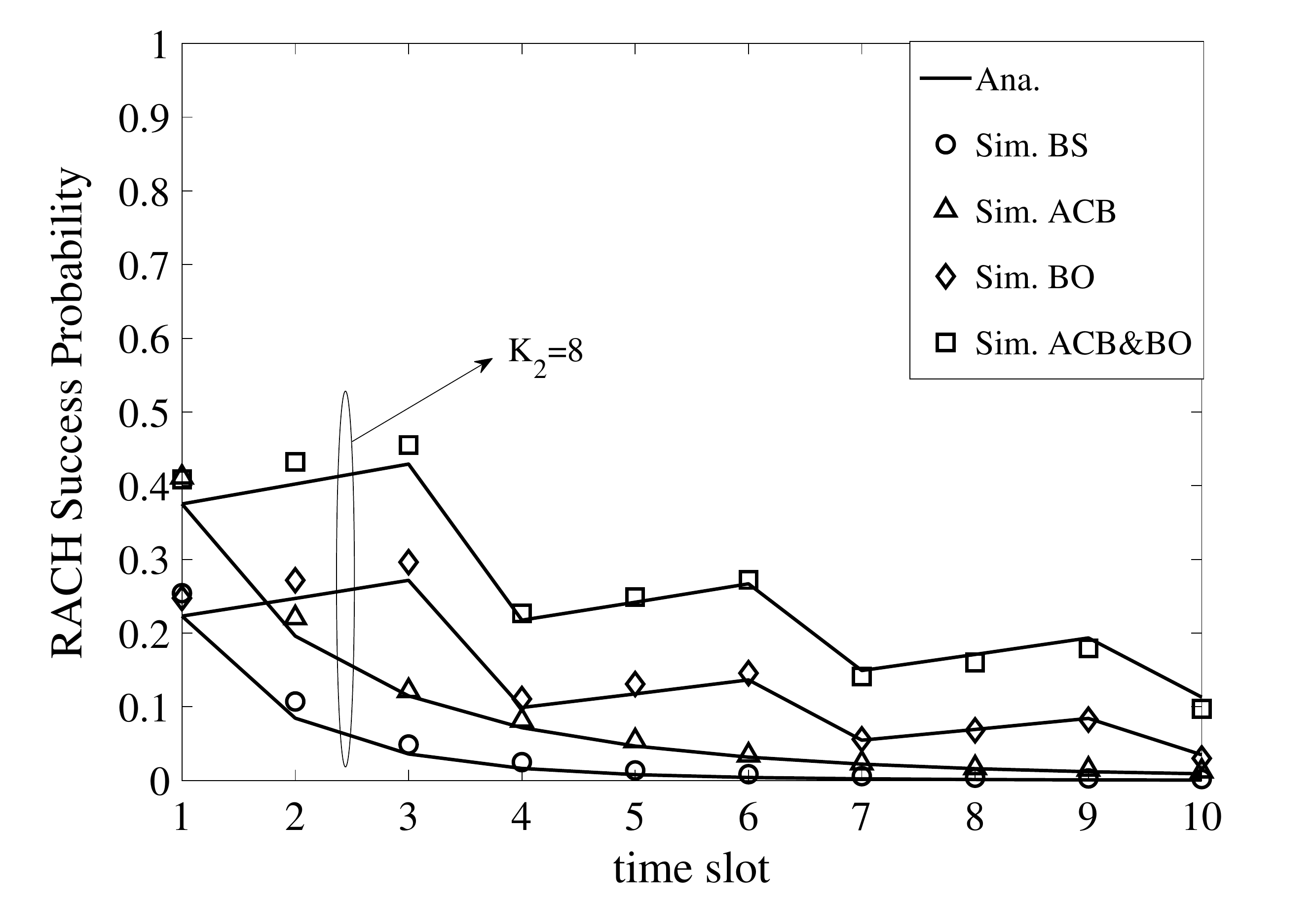}}
		\caption{ RACH success probability for three CE groups in each time slot with four RACH schemes }
		\label{fig5}
	\end{figure*}
	{We also observe that when the density ratio $\lambda_D/\lambda_B$ increases, it has the most impact on the CE group 2 and the least impact on the CE group 0, which reveals that configuring more resources for CE group 2 will ensure the massive connectivity in the NB-IoT networks.
		In both Fig. 3 and Fig. 5, it is obvious that increasing the repetition value leads to higher RACH success probabilities.
		However,  it should be noted that if the repetition value is overestimated (e.g., $K_1$=8 in CE group 1 in Fig. 4), the IoT device costs double resources than that with $K_1$=4, whereas the RACH success probabilities only improve 0.02, which will waste the potential resource for data transmission and lead to lower resource efficiency.
	}

	{Fig. 6 plots the RACH success probabilities
		of a random IoT device in each time slot with the baseline scheme, the ACB scheme, the BO scheme and the ACB $\&$ BO scheme for three CE groups, respectively. 
		For each scheme, the RACH success probabilities decrease with increasing time, due to that the intensity of interfering IoT devices grows with increasing non-empty probability of each IoT device, caused by the increasing average number of accumulated packets. Interestingly, we observe that the RACH success probabilities of a random IoT device for all three CE groups in each time slot  always follow the performance ACB$\&$BO ($Q_{ACB} = 0.6,  T_{BO}=2)>$  BO $>$ ACB $>$ baseline scheme (except the 1st time slot, where the BO procedure is not executed), this is because more strict congestion control schemes reduce the access requests from the side of IoT devices, which decrease the aggregate interference and collision probability.} 
	{For example, according to \eqref{ACB} and \eqref{BACKOFF}, the RACH success probabilities are lower than 70$\%$ leading to 49$\%$ IoT devices deferring their RACH attempts in the BO scheme, but the ACB scheme leads to only  40$\%$ deferring their RA attempts (i.e., $Q_{ACB}$ = 0.6), and thus the probabilities of deferring RACH attempt follows ACB $\&$ BO $>$  BO $>$ ACB $>$ Baseline.}
	

	
	
	\begin{figure}[htbp!]
		\begin{center}
				\centering
				\includegraphics[width=3.0in,height=2.3in]{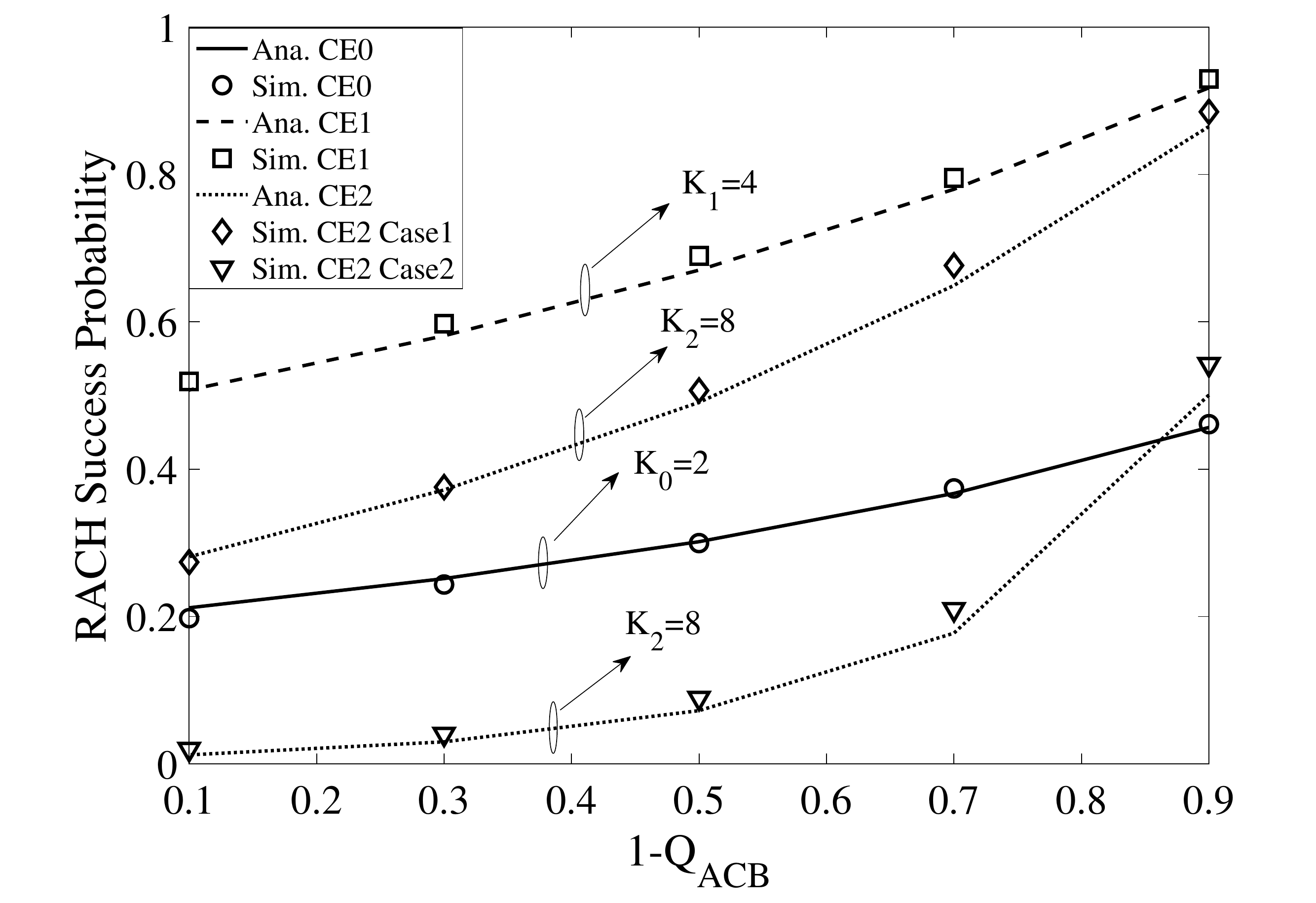}
				\vspace*{-0.1cm}
				\caption{\scriptsize RACH success probability at the 5th time slot versus ACB factors}
			\label{fig:7}
		\end{center}
	\end{figure}

	{
		Fig. 7 plots the RACH success probabilities of the ACB scheme in the 5th time slot versus the non-ACB probability 1-$Q_{ACB}$ for three CE groups, respectively. In Fig. 7, the RACH success probabilities increase with increasing 1-$Q_{ACB}$ (i.e., decreasing $Q_{ACB}$)  due to that the increasing number of IoT devices deferring access requests leads to the decrease of interference and collision probability. 
		It should be noted that the effect of the ACB scheme is more obvious in the scenario of massive connectivity, e.g, the CE group 2. }

	\begin{figure}[htbp!]
		\begin{center}
				\centering
				\includegraphics[width=3.0in,height=2.3in]{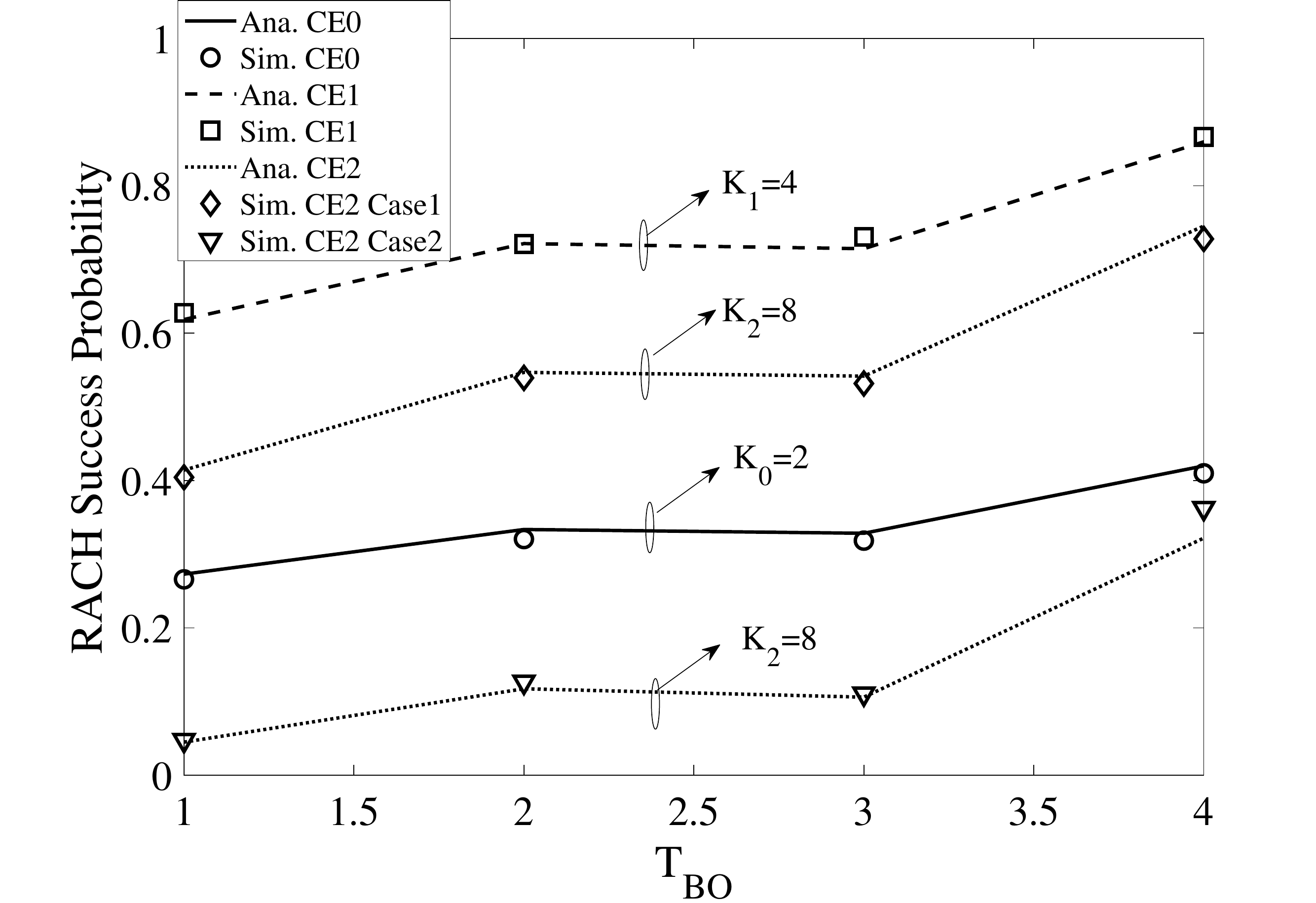}
				\vspace*{-0.1cm}
				\caption{\scriptsize RACH success probability in the 5th time slot versus BO factors}
			\label{fig:8}
		\end{center}
	\end{figure}
	
	{
		Fig. 8 plots the RACH success probabilities of the BO scheme in the 5th time slot versus the BO factor $T_{BO}$ for three CE groups, respectively. In Fig. 8, the RACH success probabilities increase with increasing $T_{BO}$, due to that the increasing number of IoT devices deferring access requests leads to the reduction of interference and collision probability. 
		Interestingly, the RACH success probabilities decrease a little bit from $T_{BO}=2$ to $T_{BO}=3$.
		This is due to the factor that when $T_{BO}=3$, the IoT devices failing to access in the 1st time slot reattempt the RACH in the 5th time slot after a backoff period 3 time slots, 
		which leads to the increase of the interference and collision probability in the 5th time slot. 
	}
	


	\section{Conclusion}
	In this paper, we developed a new spatio-temporal mathematical model to analyze the RACH  success probability under the repetition scheme in the NB-IoT networks with  three CE groups in each cell, where multiple  IoT devices  simultaneously  start their  RACH  procedure.
	We first obtained the approximate characterization of interference experienced by a randomly chosen IoT device in each CE group. 
	We then derived the analytical expression for the RACH success probability of the IoT device in the first time slot for each CE group taking into account both preamble transmission outage and collision. 
	Next, we extended the RACH success probability analysis for three CE groups to multiple time slots by modelling the queue evolution with the Baseline, Access Class Barring, Back-Off and hybrid ACB and BO schemes.
	{
		Our numerical results have shown that 
		1) the RACH success probabilities of the devices in three CE groups  outperform that of a single CE group network (almost two times);
		2)  categorizing the IoT devices into three CE groups is not always beneficial to all the groups,   which is affected by the choice of the categorizing parameters;
		3) the impact of increasing repetition value on the RACH access probabilities of CE group 1 is not so much;
		4) the RACH success probabilities follows ACB$\&$BO $>$  BO $>$ ACB $>$ baseline scheme.}
	
	
	

	\appendices
	\numberwithin{equation}{section}
	
	\section{a proof of lemma 1 }
	In our approximation approach, the baseline PPP $\Phi_{\rm{D}}$ is independently thinned such that the resulting densities of the PPPs in each CE group are the same as PHP ${\Phi _i}$, which we denote by $\lambda_i$. Then we need to derive ${\Phi _i}$ in terms of the given system parameters. For completeness, we discuss its proof briefly below. To derive ${\Phi _i}$, we need to derive an expression for the average number of devices of the ${\Phi _i}$ lying in each CE group.
	Firstly, we have the region covered by CE group 0 as 
	\begin{align}
	{\Xi _{{D_0}}} \triangleq \bigcup\limits_{y \in {\Phi _B}} b (y,{D_0}), 
	b(y,{D_0}) \equiv \{ z \in {\mathbb{R}^2}:\left\| {z - y} \right\| < {D_0}\}  
	\end{align}
	Then, the points of $\Phi_{\rm{D}}$ lying in ${\Xi _{{D_0}}}$, form ${\Phi _0}$ :
	\begin{align}
	{\Phi _0} = \{ x \in {\Phi _D}:x \in {\Xi _{{D_1}}}\} 
	\end{align}
	So the average number of points of the ${\Phi _0}$ lying in a given set ${\rm A} \subset {\mathbb{R}^2}$ is
	\begin{align}
	&{\rm E} \Big[ {\sum\limits_{x \in {\Phi _D} \cap {\rm A}} {\prod\limits_{y \in {\Phi _B}} {{\textbf{1}_{b(x,{D_0})}}(y)} } } \Big] \nonumber\\
&	\overset{\text{(a)}}= {{\rm E}_{{\Phi _D}}}\Big[ {\sum\limits_{x \in {\Phi _D} \cap {\rm A}} {{{\rm E}_{{\Phi _B}}}\Big[ {\prod\limits_{y \in {\Phi _B}} {{\textbf{1}_{b(x,{D_0})}}(y)} } \Big]} } \Big] \nonumber\\
	&\overset{\text{(b)}}= {{\rm E}_{{\Phi _D}}}\Big[ {\sum\limits_{x \in {\Phi _D} \cap {\rm A}} {\exp ( - {\lambda _B}\int\limits_{{\mathbb{R}^2}} {(1 - ({\textbf{1}_{b(x,{D_0})}}(y))dy} } } \Big] \nonumber\\
&	\overset{\text{(c)}}= \Big| {\rm A} \Big|{\lambda _D}(1 - \exp ( - {\lambda _B}\pi D_0^2)),  
	\end{align}
	where (a) is due to the independence of point processes $\Phi _B$ and $\Phi _D$, (b) follows from the probability generating functional (PGFL) of a PPP, and (c) follows from the Campbell theorem \cite{chiu2013stochastic}. From the above expression, we can readily infer
	that ${\lambda _0} = {\lambda _D}(1 - \exp ( - {\lambda _B}\pi D_0^2))$. Similarly, we can derive $\lambda_1$ and $\lambda_2$ and prove Lemma 1.

	\section{a proof of lemma 2 } 
	The Laplace transform of the aggregate interference  received at the typical BS is derived as 
	{
		\begin{align}\label{B1}
		&{\mathbb E}\Big[\exp \Big( - \frac{{{\gamma _{th}}}}{\rho }\sum\limits_{\beta  = 1}^{{l_0}} {I_{0}^{\beta}\Big) } \Big]\nonumber\\
		&= {\mathbb E}\Big[\exp \Big( - \frac{{{\gamma _{th}}}}{\rho }\sum\limits_{j \in {\mathcal {Z}_0^{}} }^{{}} {P_{0,j} \sum\limits_{\beta  = 1}^{{l_0}} {h_{0,j}^\beta }{(r_{0,j})}^{ - \alpha } } \Big)\Big]\nonumber\\
		& \overset{\text{(a)}}={\mathbb E}\Big[ {{{\prod\limits_{j \in {\mathcal {Z}_0^{}}} {\Big( {\frac{1}{{1 + P_{0,j}{(r_{0,j})}^{ - \alpha }{\gamma_{th}/\rho}}}} \Big)} }^{{l_0}}}} \Big]\nonumber\\
		& \overset{\text{(b)}}= \exp \Big( - 2\pi {\mathcal{A}_0^1\mathcal{R}_0^1\lambda _{0}^a}\int_{{(\frac{P}{\rho})}^{\frac{1}{\alpha}}}^{{\infty}}{\mathbb E}_P {\Big[ {1 - {{\Big( {\frac{1}{{1 + P{x}^{ - \alpha }{\gamma_{th}/\rho}}}} \Big)}^{{l_0}}}} \Big]} xdx\Big)\nonumber\\
		&\overset{\text{(c)}}= \exp \Big( - 2\pi {\mathcal{A}_0^1\mathcal{R}_0^1\lambda _{0}^a}(\frac{{\gamma_{th}}}{\rho})^{\frac{2}{\alpha}}{\mathbb E}[P^{\frac{2}{\alpha}}]\int_{({\gamma_{th}})^{\frac{-1}{\alpha}}}^{{\infty}}{\Big[ {1 - {{\Big( {\frac{1}{{1 + y^{ - \alpha }}}} \Big)}^{{l_0}}}} \Big]} ydy\Big)
		\end{align}}
	where  (a) is obtained by taking the average with respect to $h_{0,j}^\beta$, (b) follows from the probability generation functional (PGFL) of the PPP 
	and (c) follows by changing the variables $y =\displaystyle\frac{x}{(\gamma_{th}P/\rho)^{\frac{1}{\alpha}}}$. 
	{
		The moments of the transmit power is given as \cite{6786498}
		\begin{align}\label{moment}
		{\mathbb E}[{P^\frac{2}{\alpha}}] = \frac{{{\rho ^\frac{2}{\alpha} }\gamma \Big( {2,\pi {\lambda _B}{{( {\frac{P}{\rho }} )}^{\frac{2}{\alpha }}}} \Big)}}{{{{\pi {\lambda _B}}}\Big( {1 - \exp \big( - \pi {\lambda _B}{{( {\frac{P}{\rho }} )}^{\frac{2}{\alpha }}}\big)} \Big)}}, 
		\end{align}
		where $\gamma (a,b) = \int_0^b {{t^{a - 1}}{e^{ - t}}dt}$ is the lower incomplete gamma function. Substituting \eqref{moment} into \eqref{B1}, the final expression in Lemma 2 is derived.
	}

	\section{a proof of lemma 3 } 
	We note that the preamble transmission success probability in  \eqref{SI} depends on the transmission distances.
	According to  \eqref{distance}, we have
	\begin{align}\label{conditionr}
	&p({\gamma _{th}}) = {\mathbb E_{{R}}}\Big[{\mathbb{P}}_{i,0}[ {{\theta _1},{\theta _2},...,{\theta _{{{k_i}}}}| {{r}} } ]\Big]\nonumber\\
	&= \int\limits_{{D_{i - 1}}}^{{D_i}} {{\mathbb{P}}_{i,0}[ {{\theta _1},{\theta _2},...,{\theta _{{{k_i}}}}| {{r}} } ]} \frac{{2{r}}}{{D_{_i}^2 - D_{_{i - 1}}^2}}d{r}.  
	\end{align}

	Same as Appendix A, we have 
	\begin{align}\label{ce2distance}
	&{\mathbb P}_{i,0}[ {{\theta _1},{\theta _2},...,{\theta _{{{k_i}}}}| r} ]\nonumber\\
&	= \exp\Big( -  {{\frac{{{l_i\gamma _{th}\sigma ^2}r^\alpha }}{P}}}\Big){\mathbb E}\Big[ {\exp \Big(- {{\frac{{{\gamma _{th}}r^\alpha }}{P}}}\sum\limits_{\beta  = 1}^{{l_i}} {I_{i,0}^{\beta} \Big)\Big| r } } \Big].
	\end{align}
	
	The Laplace Transform of the aggregate interference  in CE group $i$  is obtained as
	{
		\begin{align}\label{laplacece2}
		&{\mathbb E}\Big[ {\exp \Big( {-  {{\frac{{{l_i\gamma _{th}\sigma ^2}r^\alpha }}{P}}}}\sum\limits_{\beta  = 1}^{{l_i}} {I_{i,0}^{\beta} \Big)\Big| r} } \Big]\nonumber\\
	&	= {\mathbb E}\Big[ {\exp \Big( - {\gamma _{th}r^\alpha}\sum\limits_{j \in  {\mathcal Z}_i} {\sum\limits_{\beta  = 1}^{{l_i}} {h_{i,j}^\beta } {{ {({r_{i,j}})} }^{ - \alpha }}\Big| r } } \Big]\nonumber\\
		&\overset{\text{(a)}}= {\mathbb E}\Big[ {{{\prod\limits_{{j} \in {\Phi _i}} {\Big( {\frac{1}{{1 + \gamma _{th}r^\alpha{(r_{i,j})}^{ - \alpha }}}} \Big)} }^{{l_i}}}} \Big] \overset{\text{(b)}}\nonumber\\
	&	= \exp \Big( - 2\pi {\mathcal{A}_i^1\mathcal{R}_i^1\lambda _{i}^a}\int\limits_{{D_{i}}}^{{\infty}} {\Big( {1 - {{\Big( {\frac{1}{{1 + \gamma _{th}r^\alpha{{y}^{ - \alpha }}}}} \Big)}^{{l_i}}}} \Big)} ydy\Big).
		\end{align}}
	Combing  \eqref{conditionr} -- \eqref{laplacece2}, we proved Lemma 3.

	\bibliographystyle{IEEEtran}

	\bibliography{IEEEabrv,group}

\begin{thebibliography}{10}
\providecommand{\url}[1]{#1}
\csname url@samestyle\endcsname
\providecommand{\newblock}{\relax}
\providecommand{\bibinfo}[2]{#2}
\providecommand{\BIBentrySTDinterwordspacing}{\spaceskip=0pt\relax}
\providecommand{\BIBentryALTinterwordstretchfactor}{4}
\providecommand{\BIBentryALTinterwordspacing}{\spaceskip=\fontdimen2\font plus
\BIBentryALTinterwordstretchfactor\fontdimen3\font minus
  \fontdimen4\font\relax}
\providecommand{\BIBforeignlanguage}[2]{{%
\expandafter\ifx\csname l@#1\endcsname\relax
\typeout{** WARNING: IEEEtran.bst: No hyphenation pattern has been}%
\typeout{** loaded for the language `#1'. Using the pattern for}%
\typeout{** the default language instead.}%
\else
\language=\csname l@#1\endcsname
\fi
#2}}
\providecommand{\BIBdecl}{\relax}
\BIBdecl

\bibitem{9013330}
Y.~{Liu}, Y.~{Deng}, M.~{Elkashlan}, and A.~{Nallanathan}, ``Random access
  performance for three coverage enhancement groups in {NB-IoT} networks,'' in
  \emph{2019 IEEE Global Commun. Conf. (GLOBECOM)}, Dec. 2019, pp. 1--6.

\bibitem{7123563}
A.~Al-Fuqaha, M.~Guizani, M.~Mohammadi, M.~Aledhari, and M.~Ayyash, ``Internet
  of things: A survey on enabling technologies, protocols, and applications,''
  \emph{IEEE Commun. Surveys Tuts.}, vol.~17, no.~4, Jun. 2015.

\bibitem{8058399}
S.~Vashi, J.~Ram, J.~Modi, S.~Verma, and C.~Prakash, ``Internet of things
  ({IoT}): A vision, architectural elements, and security issues,'' in
  \emph{2017 Int. Conf. on IoT in Social, Mobile, Analytics and Cloud
  (I-SMAC)}, Feb. 2017, pp. 492--496.

\bibitem{name2015}
``Cellular system support for ultra-low complexity and low throughput
  {I}nternet of {T}hings {(CIoT)},'' \emph{3GPP, Sophia Antipolis, France, TR
  45.820 V13.1.0,}, Nov. 2015.

\bibitem{6375894}
A.~Ksentini, Y.~Hadjadj-Aoul, and T.~Taleb, ``Cellular-based
  machine-to-machine: overload control,'' \emph{IEEE Netw.}, vol.~26, no.~6,
  pp. 54--60, Nov. 2012.

\bibitem{shariatmadari2015machine}
H.~Shariatmadari, R.~Ratasuk, S.~Iraji, A.~Laya, T.~Taleb, R.~Jäntti, and
  A.~Ghosh, ``Machine-type communications: current status and future
  perspectives toward 5{G} systems,'' \emph{IEEE Commun. Mag.}, vol.~53, no.~9,
  pp. 10--17, Sep. 2015.

\bibitem{landstrom2016nb}
S.~Landstr{\"o}m, J.~Bergstr{\"o}m, E.~Westerberg, and D.~Hammarwall,
  ``{NB-IoT}: A sustainable technology for connecting billions of devices,''
  \emph{Ericsson Technol. Review}, vol.~4, pp. 2--11, Apr. 2016.

\bibitem{schlienz2016narrowband}
J.~Schlienz and D.~Raddino, ``Narrowband internet of things whitepaper,''
  \emph{IEEE Microwave Mag.}, vol.~8, no.~1, pp. 76--82, Aug. 2016.

\bibitem{el2017m2m}
A.~El~Mahjoubi, T.~Mazri, and N.~Hmina, ``{M2M} and e{MTC} communications via
  {NB-IoT}, morocco first testbed experimental results and rf deployment
  scenario: New approach to improve main {5G} {KPI}s and performances,'' in
  \emph{2017 Int. Conf. on Wireless Netw. and Mobile Commun. (WINCOM)}.\hskip
  1em plus 0.5em minus 0.4em\relax IEEE, Nov. 2017, pp. 1--6.

\bibitem{huawei2018}
``{3GPP}’s low-power wide-area {I}o{T} solutions: {NB}-{I}o{T} and e{MTC},''
  \emph{Workshop on 3GPP Submission Towards IMT-2020, Brussels, Belgium}, Oct.
  2018.

\bibitem{7794567}
N.~Mangalvedhe, R.~Ratasuk, and A.~Ghosh, ``{NB-IoT} deployment study for low
  power wide area cellular {IoT},'' in \emph{IEEE 27th Annual Int. Symp. Pers.
  Indoor Mobile Radio Commun. (PIMRC)}, Sep. 2016, pp. 1--6.

\bibitem{8385556}
L.~Feltrin, M.~Condoluci, T.~Mahmoodi, M.~Dohler, and R.~Verdone, ``{NB-IoT}:
  Performance estimation and optimal configuration,'' in \emph{European
  Wireless 2018; 24th European Wireless Conf.}, May 2018, pp. 1--6.

\bibitem{dahlman20134g}
E.~Dahlman, S.~Parkvall, and J.~Skold, \emph{4G: LTE/LTE-advanced for mobile
  broadband}.\hskip 1em plus 0.5em minus 0.4em\relax Academic press, Oct. 2013.

\bibitem{luo1999stability}
W.~Luo and A.~Ephremides, ``Stability of {N} interacting queues in
  random-access systems,'' \emph{IEEE Trans. on Inf. Theory}, vol.~45, no.~5,
  pp. 1579--1587, Jul. 1999.

\bibitem{duan2013dynamic}
S.~Duan, V.~Shah-Mansouri, and V.~W.~S. Wong, ``Dynamic access class barring
  for {M2M} communications in {LTE} networks,'' in \emph{2013 IEEE Global
  Commun. Conf. (GLOBECOM),}, Dec. 2013, pp. 4747--4752.

\bibitem{7486114}
Y.~Zhong, M.~Haenggi, T.~Q.~S. Quek, and W.~Zhang, ``On the stability of static
  poisson networks under random access,'' \emph{IEEE Trans. Commun.}, vol.~64,
  no.~7, pp. 2985--2998, Jul. 2016.

\bibitem{8385148}
J.~{Yuan}, A.~{Huang}, H.~{Shan}, T.~Q.~S. {Quek}, and G.~{Yu}, ``Design and
  analysis of random access for standalone {LTE-U} systems,'' \emph{IEEE Trans.
  on Veh. Technol.}, vol.~67, no.~10, pp. 9347--9361, Oct. 2018.

\bibitem{8605340}
R.~{Harwahyu}, R.~{Cheng}, W.~{Tsai}, J.~{Hwang}, and G.~{Bianchi},
  ``Repetitions versus retransmissions: Tradeoff in configuring {NB-IoT} random
  access channels,'' \emph{IEEE Internet of Things Journal}, vol.~6, no.~2, pp.
  3796--3805, Apr. 2019.

\bibitem{6678832}
A.~Laya, L.~Alonso, and J.~Alonso-Zarate, ``Is the random access channel of lte
  and lte-a suitable for {M2M} communications? a survey of alternatives,''
  \emph{IEEE Commun. Surveys Tutorials}, vol.~16, no.~1, pp. 4--16, Jan. 2014.

\bibitem{8239592}
R.~{Harwahyu}, R.~{Cheng}, C.~{Wei}, and R.~F. {Sari}, ``Optimization of random
  access channel in {NB-IoT},'' \emph{IEEE Internet of Things Journal}, vol.~5,
  no.~1, pp. 391--402, Feb. 2018.

\bibitem{8258982}
N.~Jiang, Y.~Deng, M.~Condoluci, W.~Guo, A.~Nallanathan, and M.~Dohler,
  ``{RACH} preamble repetition in {NB-IoT} network,'' \emph{IEEE Commun.
  Lett.}, vol.~22, no.~6, pp. 1244--1247, Jun. 2018.

\bibitem{nan2018random}
N.~Jiang, Y.~Deng, X.~Kang, and A.~Nallanathan, ``Random access analysis for
  massive {IoT} networks under a new spatio-temporal model: A stochastic
  geometry approach,'' \emph{IEEE Trans. on Commun.}, pp. 1--1, Jul. 2018.

\bibitem{nan2018collision}
N.~Jiang, Y.~Deng, A.~Nallanathan, X.~Kang, and T.~Q.~S. Quek, ``Analyzing
  random access collisions in massive {IoT} networks,'' \emph{IEEE Trans. on
  Wireless Commun.}, vol.~17, no.~10, pp. 6853--6870, Oct. 2018.

\bibitem{Tel2016Physical}
``Evolved universal terrestrial radio access ({E-UTRA}): Physical layer
  procedures,'' \emph{3GPP, TS 36.213 v.13.2.0, Release 13}, Aug. 2016.

\bibitem{3gpp2011study}
``Study on {RAN} improvements for machine-type communications,'' \emph{3GPP TR
  37.868 v.11.2.0}, Sept. 2011.

\bibitem{6525600}
M.~{Hasan}, E.~{Hossain}, and D.~{Niyato}, ``Random access for
  machine-to-machine communication in {LTE}-advanced networks: issues and
  approaches,'' \emph{IEEE Commun. Mag.}, vol.~51, no.~6, pp. 86--93, Jun.
  2013.

\bibitem{andrews2011tractable}
J.~G. Andrews, F.~Baccelli, and R.~K. Ganti, ``A tractable approach to coverage
  and rate in cellular networks,'' \emph{IEEE Trans. Commun.}, vol.~59, no.~11,
  pp. 3122--3134, Nov. 2011.

\bibitem{deng2016physical}
Y.~Deng, L.~Wang, M.~Elkashlan, A.~Nallanathan, and R.~K. Mallik, ``Physical
  layer security in three-tier wireless sensor networks: A stochastic geometry
  approach,'' \emph{IEEE Trans. Inf. Forensics Security}, vol.~11, no.~6, pp.
  1128--1138, Jun. 2016.

\bibitem{deng2016artificial}
Y.~Deng, L.~Wang, S.~A.~R. Zaidi, J.~Yuan, and M.~Elkashlan, ``Artificial-noise
  aided secure transmission in large scale spectrum sharing networks,''
  \emph{IEEE Trans. Commun.}, vol.~64, no.~5, pp. 2116--2129, May 2016.

\bibitem{elsawy2013stochastic}
H.~ElSawy, E.~Hossain, and M.~Haenggi, ``Stochastic geometry for modeling,
  analysis, and design of multi-tier and cognitive cellular wireless networks:
  A survey,'' \emph{IEEE Commun. Surveys Tuts.}, vol.~15, no.~3, pp. 996--1019,
  Jun. 2013.

\bibitem{7218400}
M.~{Ahmadi}, F.~{Tong}, L.~{Zheng}, and J.~{Pan}, ``Performance analysis for
  two-tier cellular systems based on probabilistic distance models,'' in
  \emph{2015 IEEE Conference Computer Commun. (INFOCOM)}, May. 2015, pp.
  352--360.

\bibitem{7886285}
Y.~Zhong, T.~Q.~S. Quek, and X.~Ge, ``Heterogeneous cellular networks with
  spatio-temporal traffic: Delay analysis and scheduling,'' \emph{IEEE J. Sel.
  Areas Commun.}, vol.~35, no.~6, pp. 1373--1386, Jun. 2017.

\bibitem{6477828}
A.~G. Gotsis, A.~S. Lioumpas, and A.~Alexiou, ``Evolution of packet scheduling
  for machine-type communications over {LTE}: Algorithmic design and
  performance analysis,'' in \emph{IEEE Global Commun. Conf. (GLOBECOM)
  Workshops}, Dec. 2012, pp. 1620--1625.

\bibitem{chiu2013stochastic}
S.~N. Chiu, D.~Stoyan, W.~S. Kendall, and J.~Mecke, \emph{Stochastic geometry
  and its applications}.\hskip 1em plus 0.5em minus 0.4em\relax John Wiley \&
  Sons, 2013.

\bibitem{7510814}
I.~{Leyva-Mayorga}, L.~{Tello-Oquendo}, V.~{Pla}, J.~{Martinez-Bauset}, and
  V.~{Casares-Giner}, ``Performance analysis of access class barring for
  handling massive {M2M} traffic in {LTE-A} networks,'' in \emph{2016 IEEE Int.
  Conf. Commun. (ICC)}, May 2016, pp. 1--6.

\bibitem{Tel2016}
``Medium access control ({MAC}) protocol specification,'' \emph{3GPP, TS 36.321
  v.13.2.0}, Jun. 2016.

\bibitem{7785170}
R.~Ratasuk, N.~Mangalvedhe, Y.~Zhang, M.~Robert, and J.~Koskinen, ``Overview of
  narrowband {IoT} in {LTE Rel-13},'' in \emph{IEEE Conf. Stand. Commun. Netw.
  (CSCN)}, Oct. 2016, pp. 1--7.

\bibitem{7569029}
X.~{Lin}, A.~{Adhikary}, and Y.~. {Eric Wang}, ``Random access preamble design
  and detection for {3GPP} narrowband {I}o{T} systems,'' \emph{IEEE Wireless
  Commun. Lett.}, vol.~5, no.~6, pp. 640--643, Dec. 2016.

\bibitem{6516885}
T.~D. Novlan, H.~S. Dhillon, and J.~G. Andrews, ``Analytical modeling of uplink
  cellular networks,'' \emph{IEEE Trans. Wireless Commun.}, vol.~12, no.~6, pp.
  2669--2679, Jun. 2013.

\bibitem{7917340}
M.~{Gharbieh}, H.~{ElSawy}, A.~{Bader}, and M.~{Alouini}, ``Spatiotemporal
  stochastic modeling of {I}o{T} enabled cellular networks: Scalability and
  stability analysis,'' \emph{IEEE Trans. Commun.}, vol.~65, no.~8, pp.
  3585--3600, Aug. 2017.

\bibitem{7557010}
Z.~Yazdanshenasan, H.~S. Dhillon, M.~Afshang, and P.~H.~J. Chong, ``Poisson
  hole process: Theory and applications to wireless networks,'' \emph{IEEE
  Trans. Wireless Commun.}, vol.~15, no.~11, pp. 7531--7546, Nov. 2016.

\bibitem{kingman1993poisson}
J.~F.~C. Kingman, \emph{Poisson processes}.\hskip 1em plus 0.5em minus
  0.4em\relax Wiley Online Library, Jan. 1993.

\bibitem{6376046}
K.~Zhou, N.~Nikaein, and T.~Spyropoulos, ``{LTE/LTE-A} discontinuous reception
  modeling for machine type communications,'' \emph{IEEE Wireless Commun.
  Lett.}, vol.~2, no.~1, pp. 102--105, Feb. 2013.

\bibitem{gow2006mobile}
G.~Gow and R.~Smith, \emph{Mobile And Wireless Communications: An Introduction:
  An Introduction}.\hskip 1em plus 0.5em minus 0.4em\relax McGraw-Hill
  Education (UK), 2006.

\bibitem{7078932}
G.~{Lin}, S.~{Chang}, and H.~{Wei}, ``Estimation and adaptation for bursty
  {LTE} random access,'' \emph{IEEE Trans. Veh. Technol.}, vol.~65, no.~4, pp.
  2560--2577, Apr. 2016.

\bibitem{yu2013downlink}
S.~M. Yu and S.~L. Kim, ``Downlink capacity and base station density in
  cellular networks,'' in \emph{11th Int. Symp. Model. Optim. Mobile Ad Hoc
  Wireless Netw. (WiOpt)}, May 2013, pp. 119--124.

\bibitem{56383}
E.~S. Sousa and J.~A. Silvester, ``Optimum transmission ranges in a
  direct-sequence spread-spectrum multihop packet radio network,'' \emph{IEEE
  J. Select. Areas Commun.}, vol.~8, no.~5, pp. 762--771, Jun. 1990.

\bibitem{6786498}
H.~{ElSawy} and E.~{Hossain}, ``On stochastic geometry modeling of cellular
  uplink transmission with truncated channel inversion power control,''
  \emph{IEEE Trans. Wireless Commun.}, vol.~13, no.~8, pp. 4454--4469, Aug.
  2014.

\end{thebibliography}
	
\end{document}